\documentclass[twoside]{article}

\usepackage[round]{natbib}

\usepackage[utf8]{inputenc} 
\usepackage[T1]{fontenc}    
\usepackage{url}            
\usepackage{booktabs}       
\usepackage{amsfonts}       
\usepackage{nicefrac}       
\usepackage{microtype}      

\usepackage{subcaption}
\usepackage{multirow}

\usepackage{amsmath,amsthm,amssymb}

\usepackage[utf8]{inputenc} 
\usepackage[T1]{fontenc}    
\usepackage{amsfonts}       
\usepackage{nicefrac}       
\usepackage{microtype}      

\usepackage{graphicx,color}
\usepackage{subcaption}
\usepackage{wrapfig}

\usepackage{tabularx}       
\usepackage{booktabs}       

\usepackage{algorithm} 
\usepackage{algorithmicx}
\usepackage[noend]{algpseudocode}

\usepackage{natbib}

\usepackage{url}            
\usepackage{cleveref}
\usepackage{autonum}
\usepackage{mathtools}
\usepackage{tikz}
\usetikzlibrary{positioning, arrows.meta, decorations.pathmorphing}

\usepackage{braket}
\usepackage{multirow}
\usepackage{lipsum}
\allowdisplaybreaks
\usepackage{lscape}

\renewcommand{\hbar}{\overline{h}}

\newcommand{\R}{\mathbb{R}}
\newcommand{\N}{\mathbb{N}}

\newcommand{\Z}{\mathbb{Z}}

\newcommand{\ab}{\text{{\boldmath $a$}}}
\newcommand{\bb}{\text{{\boldmath $b$}}}
\newcommand{\cb}{\text{{\boldmath $c$}}}

\newcommand{\ub}{\text{{\boldmath $u$}}}

\newcommand{\wb}{\text{{\boldmath $w$}}}
\newcommand{\xb}{\text{{\boldmath $x$}}}
\newcommand{\yb}{\text{{\boldmath $y$}}}
\newcommand{\zb}{\text{{\boldmath $z$}}}

\newtheorem{thm}{Theorem}
\newtheorem{lem}{Lemma}
\newtheorem{cor}{Corollary}

\theoremstyle{definition}

\DeclareMathOperator*{\minimize}{minimize}
\DeclareMathOperator*{\maximize}{maximize}
\DeclareMathOperator{\subto}{subject\ to}
\DeclareMathOperator*{\argmax}{argmax}
\DeclareMathOperator*{\argmin}{argmin}

\newcommand{\supp}{\mathrm{supp}}

\newcommand{\ceil}[1]{\left\lceil {#1} \right\rceil}
\newcommand{\floor}[1]{\left\lfloor {#1} \right\rfloor}

\newcommand{\unif}{u}
\newcommand{\relmid}[1]{\mathrel{#1|}}

\newif\iffigure
\figuretrue

\mathcode`@="8000 
{\catcode`\@=\active\gdef@{\mkern1mu}}

\mathchardef\Gamma="7100 \mathchardef\Delta="7101
\mathchardef\Theta="7102 \mathchardef\Lambda="7103
\mathchardef\Psi="7104 \mathchardef\Pi="7105 \mathchardef\Sigma="7106
\mathchardef\Upsilon="7107 \mathchardef\Phi="7108
\mathchardef\Psi="7109 \mathchardef\Omega="710A

\newcommand{\dset}{{[d]}}

\usepackage[margin=.8in]{geometry}

\begin{document}
	
	%
	
	%
	
	
	\title{
		On Maximization of Weakly Modular Functions:
		Guarantees of Multi-stage Algorithms, Tractability, and Hardness
	}
	
	\author{Shinsaku Sakaue
	\\
NTT Communication Science Laboratories}
	
	\maketitle
	
	\begin{abstract}
	Maximization of {\it non-submodular} functions appears in various scenarios, and many previous works studied it based on some measures that quantify the closeness to being submodular. On the other hand, many practical non-submodular functions are actually close to being {\it modular}, which has been utilized in few studies. In this paper, we study cardinality-constrained maximization of {\it weakly modular} functions, whose closeness to being modular is measured by {\it submodularity} and {\it supermodularity ratios}, and reveal what we can and cannot do by using the weak modularity. We first show that guarantees of multi-stage algorithms can be proved with the weak modularity, which generalize and improve some existing results, and experiments confirm their effectiveness. We then show that weakly modular maximization is {\it fixed-parameter tractable} under certain conditions; as a byproduct, we provide a new time--accuracy trade-off for $\ell_0$-constrained minimization. We finally prove that, even if objective functions are weakly modular, no polynomial-time algorithms can improve the existing approximation guarantees achieved by the greedy algorithm. 
	\end{abstract}

\newcommand{\poly}{\mathrm{poly}}

\newcommand{\Ml}{\mathcal{M}}
\newcommand{\Il}{\mathcal{I}}
\newcommand{\Bl}{\mathcal{B}}

\renewcommand{\S}[1]{{S_{#1}}}
\renewcommand{\O}[1]{{O_{#1}}}
\renewcommand{\N}{N}

\newcommand{\cuz}[1]{\text{{$\because$~{#1}}}}

\newcommand{\curv}{\alpha}
\newcommand{\icurv}{{\check \alpha}}

\renewcommand{\k}{k}
\newcommand{\s}{s}
\newcommand{\kp}{{\k^\prime}}
\renewcommand{\sp}{{\s^\prime}}

\newcommand{\br}{\gamma}
\newcommand{\pr}{\beta}
\newcommand{\bruk}{\br_{\Us,\k}}
\newcommand{\pruk}{\pr_{\Us,\k}}
\newcommand{\brus}{\br_{\Us,\s}}
\newcommand{\prus}{\pr_{\Us,\s}}
\newcommand{\brkk}{\br_{\k}}
\newcommand{\prkk}{\pr_{\k}}

\newcommand{\bs}{\backslash}

\newcommand{\ls}{l}
\newcommand{\f}{f}
\newcommand{\F}{F}
\newcommand{\Fdel}[2]{\F({#1}\mid{#2})}
\newcommand{\G}{G}
\newcommand{\Gdel}[2]{\G({#1}\mid{#2})}
\renewcommand{\c}{c}

\renewcommand{\d}{d}
\newcommand{\As}{{\sf A}}
\newcommand{\Bs}{{\sf B}}
\newcommand{\Ss}{{\sf S}}
\newcommand{\Ts}{{\sf T}}
\newcommand{\Ls}{{\sf L}}
\newcommand{\Us}{{\sf U}}
\newcommand{\Is}{{\sf I}}
\newcommand{\Ms}{{\sf M}}

\newcommand{\return}{{\bf return}}

\renewcommand{\bb}[1]{\mathbf{b}^{({#1})}}
\newcommand{\bvec}{\mathbf{b}}
\newcommand{\bv}{\tau}

\newcommand{\iprod}[2]{\langle{#1},{#2}\rangle}

\newcommand{\Tl}{\mathcal{T}}

\newcommand{\singlerun}{{\tt SingleRun}\text{$()$}}
\newcommand{\Sbest}{\Ss_{\text{best}}}

\newcommand{\gbk}{{\brkk}{\pr_{k,\d}}}
\newcommand{\bgk}{\frac{1}{\brkk\pr_{k,\d}}}

\newcommand{\I}{\mathbf{I}}
\newcommand{\A}{\mathbf{A}}
\newcommand{\B}{\mathbf{B}}
\newcommand{\X}{\mathbf{X}}
\renewcommand{\H}{\mathbf{H}}

\newcommand{\m}[1]{\mu_{{#1}}}
\newcommand{\M}[1]{\nu_{{#1}}}

\newcommand{\tm}[1]{\tilde{\mu}_{{#1}}}
\newcommand{\tM}[1]{\tilde{\nu}_{{#1}}}

\newcommand{\Mst}{\M{\Ss,\Ss\cup \Ts}}
\newcommand{\mst}{\m{\Ss,\Ss\cup \Ts}}

\newcommand{\Mab}{\M{\As,\As\cup \Bs}}
\newcommand{\mab}{\m{\As,\As\cup \Bs}}

\newcommand{\tMab}{\tM{\As,\As\cup \Bs}}

\newcommand{\km}[2]{\m{{#1},{#2}}}
\newcommand{\kM}[2]{\M{{#1},{#2}}}

\newcommand{\co}{\c^*}
\newcommand{\ct}{\tilde{\c}}

\newcommand{\kt}{{\tilde{\k}}}
\newcommand{\ko}{{\k^*}}
\newcommand{\tk}{t}

\newcommand{\Pl}{\mathcal{P}}
\newcommand{\prj}[1]{\Pl_k({#1})}

\newcommand{\bek}{{\pr_{\emptyset,\k}}}
\newcommand{\beko}{{\pr_{\emptyset,\ko}}}
\newcommand{\Sso}{\Ss^*}
\newcommand{\Sst}{\Ss_{t}}
\newcommand{\Sstt}{\Ss_{t+1}}
\newcommand{\mkkk}{\m{2k+k^*}}
\newcommand{\Mkkk}{\M{2k+k^*}}
\newcommand{\teps}{\tilde{\epsilon}}

\newcommand{\Fl}{\mathcal{F}}

\newcommand{\ip}{{i^\prime}}
\newcommand{\ar}{\alpha}

\newcommand{\acc}{x}

\newcommand{\zbh}{{\hat \zb}}

\newcommand{\XS}{\X_{\Ss}}
\newcommand{\XSS}{\X_{\Ss,\Ss}}
\newcommand{\AS}{\A_{\Ss}}
\newcommand{\ASS}{\A_{\Ss,\Ss}}
\newcommand{\ASSj}{\A_{{\Ss\cup\{j\}},\Ss\cup\{j\}}}

\newcommand{\NI}{\text{NI}}

\newcommand{\xtrue}{\xb_{\rm true}}
\newcommand{\ytrue}{\yb_{\rm true}}
\newcommand{\ynoise}{\yb_{\rm noise}}
\renewcommand{\unif}{\mathcal{N}}

\newcommand{\Tfpt}{\left\lceil\left(\bgk\cdot
	\frac{\tilde{\F}+\epsilon}{\epsilon}\right)^k
	\log\delta^{-1}\right\rceil}

\newcommand{\bok}{\pr_{\emptyset,\ko}}
\newcommand{\mkko}{\m{\k+\ko}}
\newcommand{\Mkoo}{\M{\k+1,1}}
\newcommand{\gammM}{\frac{\mkko}{\Mkoo}}
\newcommand{\gamMm}{\frac{\Mkoo}{\mkko}}
\newcommand{\emcb}{\exp\left( {-\gammM\cdot\frac{\c}{\co\bok}}\right)}

\newcommand{\St}{{\tilde \Ss}}
\newcommand{\Sh}{{\hat \Ss}}
\newcommand{\tkh}{{\hat t}}
\newcommand{\Ft}{{\tilde \F}}
\newcommand{\Ftdel}[2]{\Ft({#1}\mid{#2})}
\newcommand{\Gt}{{\tilde \G}}
\newcommand{\Gtdel}[2]{\Gt({#1}\mid{#2})}
\newcommand{\GSt}{{\G(\St)}}

\newcommand{\bmin}{b_{\min}}
\newcommand{\bmax}{b_{\max}}
\newcommand{\mm}{m}
\newcommand{\ave}[1]{{\bar {#1}}}
\newcommand{\app}{\theta}

\newcommand{\Gsim}{\G_{\Ss_{i-1}}}
\newcommand{\Gs}{G_{\Ss}}
\newcommand{\Gsdel}[2]{\Gs({#1}\mid{#2})}

\newcommand{\Os}{{\sf O}}
\newcommand{\ms}{{m_{\sf S}}}
\newcommand{\ns}{{n_{\sf S}}}
\newcommand{\mt}{{m_{\sf T}}}
\newcommand{\nt}{{n_{\sf T}}}
\newcommand{\ml}{{m_{\sf L}}}
\newcommand{\nl}{{n_{\sf L}}}
\newcommand{\msl}{{m_{{\sf S}\cup{\sf L}}}}
\newcommand{\nsl}{{n_{{\sf S}\cup{\sf L}}}}
\newcommand{\ma}{{m_{\sf A}}}
\newcommand{\na}{{n_{\sf A}}}
\newcommand{\nb}{{n_{\sf B}}}

\newcommand{\Go}{G_1^k}
\newcommand{\Godel}[2]{\Go({#1}\mid{#2})}

\newcommand{\Gr}{G_r^k}
\newcommand{\Grdel}[2]{\Gr({#1}\mid{#2})}

\newcommand{\Gor}{G_1^{k-r+1}}
\newcommand{\Gordel}[2]{\Gor({#1}\mid{#2})}

\newcommand{\Do}[2]{\Delta_1^k({#1}\mid{#2})}
\newcommand{\Dr}[2]{\Delta_r^k({#1}\mid{#2})}
\newcommand{\Dor}[2]{\Delta_1^{k-r+1}({#1}\mid{#2})}

\newcommand{\Hl}{{H^\ell}}
\newcommand{\Hldel}[2]{{\Hl({#1}\mid{#2})}}

\newcommand{\Hr}{{H^{k-r+1}}}
\newcommand{\Hrdel}[2]{{\Hr({#1}\mid{#2})}}

\newcommand{\Rl}{{R^\ell}}

\newcommand{\pk}{p^k_r}
\newcommand{\prlb}{\left(2+\frac{r-1}{k-r+1} \right)^{-1}}

\newcommand{\fpt}[1]{{\tt Rand-FPT-Approx}{#1}}
\newcommand{\greedy}[1]{{\tt Greedy}{#1}}
\newcommand{\mgreedy}[1]{{\tt Multi-Greedy{#1}}}
\newcommand{\omp}[1]{{\tt OMP}{#1}}
\newcommand{\momp}[1]{{\tt Multi-OMP{#1}}}
\newcommand{\iht}[1]{{\tt IHT}{#1}}
\newcommand{\htp}[1]{{\tt HTP}{#1}}

\newcommand{\zeros}{\mathbf{0}}
\newcommand{\ones}{\mathbf{1}}
\newcommand{\random}{{\tt Random}}

\newcommand{\Omegarm}{\mathrm{\Omega}}
\newcommand{\Thetarm}{\mathrm{\Theta}}
\newcommand{\tOmegarm}{\tilde{\Omegarm}}

\section{INTRODUCTION}\label{sec:introduction}
We consider the following 
set function maximization with a cardinality constraint:  
\begin{align}
\label{prob:main_F}
\maximize_{\Ss\subseteq\dset}\ \F(\Ss) \quad \subto\ |\Ss|\le k, 
\end{align}
where 
$\d,\k\in\Z_{>0}$, 
$\dset\coloneqq\{1,\dots,\d\}$,  
and 
$\F:2^\dset\to\R$. 
We assume $\F$ to be monotone, normalized, and {\it weakly modular} (WM), 
where the closeness to being modular is represented with 
{\it submodularity ratio} (SBR) $\br\in[0,1]$ 
and {\it supermodularity ratio} (SPR) $\pr\in[0,1]$; 
(see, \Cref{subsec:background} for precise definitions).  
We say $\F$ is {\it weakly submodular} ({\it weakly supermodular}) 
if its SBR (SPR) is lower bounded. 
The larger SBR and SPR are, 
the closer $\F$ is to being submodular and supermodular, 
respectively, 
and $F$ is modular if $\gamma=\beta=1$.

Many previous studies on non-submodular maximization are based on 
some measures that quantify the deviation from being submodular \citep{elenberg2018restricted,qian2019fast}, 
and SBR is one of the most prevalent among such measures. 
As regards weakly submodular maximization, 
\citet{das2018approximate} proved a well-known $(1-e^{-\br})$-approximation guarantee of the greedy algorithm (\greedy{}). 

When it comes to practical non-submodular maximization instances, 
it can be effective to 
employ additional measures other than those quantifying the distance to being submodular. 
\citet{bian2017guarantees} considered a problem class such that $\F$ has bounded SBR and {\it curvature} $\curv\in[0,1]$, 
and they proved a $\frac{1}{\curv}(1-e^{-\curv\br})$-approximation guarantee 
of \greedy{}. 
Namely, 
an improved approximation guarantee is possible if $\curv<1$.  
Unfortunately, however, 
$\curv=1$ occurs quite naturally in many applications as in \Cref{sec:applications} (see also \citep{soma2018new}), 
which motivates us to 
consider a wider class of non-submodular maximization 
that can capture the structures of various practical problems.  

Weakly modular maximization (WMM) forms a wider class than that of \citep{bian2017guarantees}. 
In fact, 
SPR $\pr$ and curvature $\curv$ always satisfy 
$\pr\ge1-\curv$ \citep{bogunovic2018robust}; 
i.e., 
SPR $\beta$ and be bounded even if $\alpha=1$. 
As shown in \Cref{sec:applications}, 
various problems 
including feature selection~\citep{das2018approximate} 
and production planning~\citep{bian2017guarantees} 
strictly belong to WMM; 
that is, 
$\F$ has bounded SPR $\pr$ even though $\curv=1$ in general.  
This fact suggests the importance of studying WMM. 
However, few previous works have studied problem~\eqref{prob:main_F} 
by utilizing the weak modularity, and so WMM remains 
to be studied.

\subsection{Our Contribution}
Our first contribution provides guarantees of efficient algorithms for WMM. 
As described in \Cref{sec:applications}, 
WMM can model various continuous optimization problems including 
$\ell_0$-constrained minimization and linear programming (LP) 
with a cardinality constraint. 
Given such WMM instances, 
the evaluation of objective functions involves solving optimization subproblems, 
which is often so costly that even standard \greedy{} becomes impractical. 
To overcome this hardship, 
we consider using {\it multi-stage} algorithms for WMM.

\paragraph{Guarantees of Multi-stage Algorithms}
In \Cref{sec:multi}, we show that 
the {multi-stage} approach is effective for costly WMM instances; 
with this approach, we accelerate greedy-style algorithms 
by adding multiple elements in each iteration instead of a single element.  
The only existing study that proved guarantees of multi-stage algorithms 
is \citep{wei2014fast}; 
{their result requires the submodularity and its approximation ratio 
	is expressed as $\frac{1}{\alpha}(1-e^{-{\curv}{(1-\curv)}})$ in general, 
	which becomes $0$ if curvature $\curv$ is equal to $1$.} 
Our guarantee of 
the multi-stage greedy algorithm (\mgreedy{}) for WMM is advantageous relative to the previous result in two aspects: 
It can be applied to WM functions, which are generally non-submodular, 
and it can yield positive approximation ratios 
even if $\curv=1$ 
as long as SBR and SPR are bounded.  
Our result also includes the 
existing $(1-e^{-\br})$-approximation 
guarantee~\citep{das2018approximate} 
as a special case. 
We then focus on $\ell_0$-constrained minimization 
and prove 
a guarantee of the 
multi-stage orthogonal matching pursuit 
(\momp{}), 
which can achieve a better approximation ratio than \mgreedy{}. 
Surprisingly, 
our result 
matches that of standard \omp{} \citep{elenberg2018restricted}, 
while \momp{} can run faster than $\omp{}$. 
Our result also improves that of the latest feature selection 
algorithm \citep{qian2019fast}. 
Experiments show that 
the multi-stage approach successfully accelerates 
\greedy{} and \omp{} 
at the cost of a slight decline in solution quality. 

\paragraph{}
Our second and their contributions, 
presented in \Cref{sec:theoretical}, 
are related to theoretical properties of WMM. 
As detailed below, 
these contributions are important for revealing what we can and cannot do with the weak modularity (or bounded SBR and SPR).

\paragraph{Fixed-parameter Tractability}
In \Cref{sec:fpt}, 
we show that $\epsilon$-error solutions for WMM can be obtained
with 
a randomized {\it fixed-parameter tractable} (FPT) algorithm, 
whose computation cost 
depends arbitrarily on 
certain inputs including SBR $\br$, SPR $\pr$, sparsity $\k$, and $\epsilon$, 
but it is polynomial in $\d$. 
The algorithm we use was developed by \citet{skowron2017fpt}, 
but its guarantee was proved only for a special case of 
monotone submodular maximization. 
We also provide a time--accuracy trade-off 
for $\ell_0$-constrained minimization 
as a byproduct, 
which is contrasted with the 
existing sparsity--accuracy trade-off \citep{shalev2010trading}.

\paragraph{Hardness of Improving Approximation Ratio}
As mentioned before, 
if curvature $\curv$ is bounded by a constant smaller than $1$, 
the $\frac{1}{\curv}(1-e^{-\curv\br})$-approximation guarantee of \citep{bian2017guarantees} improves the approximation ratio, $1-e^{-\br}$, 
of \citep{das2018approximate}.  
When it comes to WMM, 
not curvature $\curv$ but SPR $\pr$ ($\ge1-\curv$) is bounded. 
Given this background, the following question arises: 
Can we improve the approximation ratio, $1-e^{-\br}$, 
if SPR $\pr$ is bounded by a constant, instead of curvature $\curv$.  
In \Cref{sec:hard}, we prove that it is generally impossible in polynomial time. 
More precisely,  we prove that, 
even if $\br$ and $\pr$ are lower bounded by some constants,  
no polynomial-time algorithms can improve the 
$(1-e^{-\br})$-approximation guarantee in general in the value oracle model. 
This result clarifies the theoretical gap between SPR $\pr$ and curvature $\curv$.

\subsection{Notation and Definitions}\label{subsec:background}  
Given any $\F:2^\dset\to\R$, 
we define $\Fdel{\Ts}{\Ss}\coloneqq\F(\Ss\cup\Ts)-\F(\Ss)$
for any $\Ss,\Ts\subseteq\dset$. 
All the set functions considered in this paper are monotone  
($\Fdel{\Ts}{\Ss}\ge0$, $\forall\Ss,\Ts\subseteq\dset$) 
and normalized 
($\F(\emptyset)=0$).
We say $\F$ is submodular (supermodular) 
if 
$\Fdel{j}{\Ss}\ge\Fdel{j}{\Ts}$ 
($\Fdel{j}{\Ss}\le\Fdel{j}{\Ts}$) 
holds for any
$\Ss\subseteq\Ts$ and $j\notin\Ts$. 
We assume that $\F$ 
can be evaluated in polynomial time w.r.t. $\d$ 
(or $\poly(\d)$ time). 
Given any $\Ss\subseteq \dset$ and $\xb\in\R^\dset$,
whose $j$-th entry $\xb_j$ is associated with $j\in\dset$,
$\xb_\Ss\in\R^\Ss$ denotes the restriction of $\xb$ to $\Ss$. 
We define 
the support of $\xb$ as 
$\supp(\xb)\coloneqq \{ j\in\dset \relmid{} \xb_j\neq0 \}$.

\paragraph{SBR and SPR} 
Given any monotone 
$\F:2^\dset\to\R$, 
$\Us\subseteq\dset$, 
and $\s\in\Z_{>0}$, 
we define SBR $\brus$ 
and 
SPR $\prus$ 
of $\F$ 
as the largest scalars that satisfy         
\begin{align}
\brus
{\Fdel{\Ss}{\Ls}}
\le
{{\sum}_{j\in \Ss}\Fdel{j}{\Ls}} 
\le 
\prus^{-1}
{\Fdel{\Ss}{\Ls}}    
\end{align}
for any disjoint $\Ls,\Ss\subseteq\dset$ 
such that $\Ls\subseteq\Us$ and $|\Ss|\le\s$.  
We say $\F$ is ($\br_{\Us_1,\s_1}$, $\pr_{\Us_2,\s_2}$)-WM 
if $\F$ has bounded $\br_{\Us_1,\s_1}$ and $\pr_{\Us_2,\s_2}$. 
Note that 
$\br_{\Us^\prime,\sp}\ge\brus$ 
and 
$\pr_{\Us^\prime,\sp}\ge\prus$ hold 
for any $\Us^\prime\subseteq \Us$ and $\sp\le \s$.
We can confirm that 
$\brus\in[0,1]$ 
and 
$\prus\in[1/\s,1]$ hold for any $\Us$ and $\s$.  
We define  $\br_{\sp,\s}\coloneqq\min_{|\Us|\le\sp}\brus$ 
and
$\pr_{\sp,\s}\coloneqq\min_{|\Us|\le\sp}\prus$; 
we sometimes use $\br_{\s}\coloneqq\br_{\s,\s}$ and $\pr_{\s}\coloneqq\pr_{\s,\s}$. 
We have 
$\br_{\d}=1$ 
($\pr_{\d}=1$) 
iff $\F$ is submodular (supermodular).

\paragraph{Curvature}
Given monotone $\F:2^\dset\to\R$, 
its curvature~$\curv\in[0,1]$ 
is defined as the smallest scalar  
that satisfies 
$
{\Fdel{j}{\Ss\bs\{j\}\cup\Ms}}
\ge
(1-\curv)
{\Fdel{j}{\Ss\bs\{j\}}} 
$
for any 
$\Ss,\Ms\subseteq\dset$ and $j\in\Ss\bs\Ms$. 
We always have  
$\prus\ge1-\curv$ for any $\Us$ and $\s$ 
(see, \citep{bogunovic2018robust}). 

\paragraph{Restricted Strong Convexity and Restricted Smoothness} 
When studying $\ell_0$-constrained minimization algorithms, 
the restricted strong convexity (RSC) and restricted smoothness (RSM) 
of loss function $\ls:\R^\d\to\R$ is often used~\citep{jain2014iterative,elenberg2018restricted,yuan2018gradient}. 
We assume $\ls$ to be differentiable. 
Given any fixed $\s_1,\s_2\in\Z_{>0}$, 
we say $\ls$ is   
$\km{\s_1}{\s_2}$-RSC 
and 
$\kM{\s_1}{\s_2}$-RSM 
if it satisfies
\begin{align}
\ls(\yb) 
\ge 
\ls(\xb) 
+
\iprod{\nabla l(\xb)}{\yb-\xb}
+
\frac{\km{\s_1}{\s_2}}{2}\|\yb-\xb\|_2^2
\quad
\text{and}
\quad
\ls(\yb)  
\le
\ls(\xb) 
+
\iprod{\nabla l(\xb)}{\yb-\xb}
+
\frac{\kM{\s_1}{\s_2}}{2}\|\yb-\xb\|_2^2, 
\end{align}
respectively, 
for any $\xb,\yb\in\R^\d$ 
such that
$\|\xb\|_0\le \s_1$, 
$\|\yb\|_0\le \s_1$, and 
$\|\xb-\yb\|_0\le \s_2$. 
If $\ls$ is quadratic, 
the above inequalities 
reduce to 
those of 
the well-known {restricted isometric property} 
(RIP) condition~\citep{candes2006stable}.
We let 
$\m{\s}\coloneqq\km{\s}{\s}$
and 
$\M{\s}\coloneqq\kM{\s}{\s}$. 
We define the restricted condition number
as $\kappa_{\s}\coloneqq \M{\s}/\m{\s}$. 
Typically,  
$\ls$ with a smaller $\kappa_{\s}$ value is easier to deal with. 
If $\ls$ is $\m{\d}$-RSC and $\M{\d}$-RSM, 
we abbreviate the subscript and say $\ls$ is 
$\m{}$-strongly convex 
($\m{}$-SC) and 
$\M{}$-smooth 
($\M{}$-SM);  
we call $\kappa\coloneqq\nu/\mu$  a condition number. 

\subsection{Related Work}\label{subsec:related}
For the case where $\F$ is submodular,   
\citet{nemhauser1978analysis} 
proved the $(1-e^{-1})$-approximation guarantee of \greedy{}. 
\citet{nemhauser1978best} proved that 
no polynomial-time algorithms can improve this guarantee in the value oracle model, 
and 
\citet{feige1998threshold} proved the  
NP-hardness for the case of {\it Max $k$-cover}. 
As regards tractability, 
\citet{skowron2017fpt} developed a randomized FPT approximation algorithm for 
maximization of monotone submodular functions with a special property called 
{\it $p$-separability}. 
Unlike our results, 
those results hold only for monotone submodular maximization. 

When it comes to non-submodular maximization, 
various notions have been introduced to obtain 
theoretical guarantees 
\citep{krause2010dictionary,
	feige2013welfare,
	horel2016maximization,	
	wang2016approximation,
	zhou2016causal}. 
\citet{das2018approximate} 
proposed SBR, 
one of the most prevalent notion used in many studies 
\citep{hu2016efficient,
	elenberg2017streaming,	
	khanna2017approximation,
	khanna2017scalable,
	chen2018weakly,
	qian2019fast},  
and they proved that \greedy{} outputs solution $\Ss$ 
with a ($1-e^{-\br_{\Ss,\k}}$)-approximation guarantee. 
\citet{harshaw2019submodular} proved that 
no polynomial-time algorithms can improve 
this approximation guarantee for weakly submodular maximization in general. 
This result is different from our hardness result 
since they do not 
assume $\F$ to have SPR bounded by a constant; 
this difference is critical 
since bounded SPR could make the problem easier. 
The definition of SPR that we use was introduced by~\citet{bogunovic2018robust}. 
While other SPR-like notions
have been used in the context of minimization problems~\citep{takeda2013simultaneous,liberty2017greedy}, 
those are different from SPR, 
which quantifies the deviation from supermodularity 
in the context of maximization problems. 

Curvature $\curv$~\citep{conforti1984submodular,bian2017guarantees} 
is also used in many studies~\citep{
	iyer2013curvature,
	sviridenko2015optimal,
	bai2018greedy}. 
Its value is, however, often pessimistic (i.e., $\curv\approx1$) as pointed out by \citet{soma2018new}, 
and to bound the curvature value is more demanding than to bound SPR.  
Hence our results obtained with SPR  
are different from existing guarantees that use curvature; 
although those results can sometimes be improved 
by using {\it greedy curvature} $\curv_G\le\curv$~\citep{bian2017guarantees}, 
no lower bounds of $\curv_G$ for WM functions have been proved.

We remark that our work is different from some previous studies on set functions that are close to being modular. 
As mentioned before, \citet{bian2017guarantees} 
studied the case where the curvature and SBR are bounded, 
and they proved that \greedy{} finds solution $\Ss$ with 
a $\frac{1}{\curv}(1-e^{-\curv\br_{\Ss,\k}})$-approximation guarantee. 
They also proved that \greedy{} cannot improve this guarantee. 
Unlike this result, our hardness result considers every 
polynomial-time algorithm. 
\citet{bogunovic2018robust} considered the case where SBR and SPR are bounded. 
However, they are interested in obtaining guarantees for robust maximization, 
not for the standard cardinality-constrained maximization~\eqref{prob:main_F}, 
which is of our interest. 
\citet{chierichetti2015approximate} 
defined the approximate modularity as the $\ell_\infty$-distance to being modular, 
which is different from the weak modularity. 

\citet{wei2014fast} provided 
curvature-dependent approximation guarantees of 
multi-stage algorithms for submodular maximization. 
\citet{marsousi2013multi} 
applied \momp{} to a special case of $\ell_0$-constrained minimization 
where the loss function $\ls$ is quadratic, 
but its theoretical guarantee has not been proved. 
The idea of adding multiple elements in each round 
is also considered in the context of parallel algorithms~\citep{balkanski2018adaptive}. %
\citet{qian2019fast} have recently developed 
a parallel approximation algorithm that runs in $O(\log(\d))$ 
rounds for $\ell_0$-constrained minimization. 
Surprisingly, 
thanks to the use of the weak modularity, 
we can show that 
\momp{} with only one round 
achieves a better approximation guarantee.

\section{APPLICATIONS}\label{sec:applications}
We motivate to study WMM by presenting its applications. 
For each application, we present lower bounds of SBR and SPR. 
We also provide an example such that $\curv=1$ 
for one of the applications, 
and such examples for the other applications are presented in \Cref{a_sec:applications}.

\paragraph{$\ell_0$-constrained Minimization} 
Given a differentiable loss function $\ls:\R^\d\to\R$, 
we consider the following $\ell_0$-constrained minimization: 
$\min_{\|\xb\|_0\le\k}\ls(\xb)$. 
It is generally NP-hard \citep{natarajan1995sparse} 
and
appears in many practical scenarios: 
Feature selection \citep{das2018approximate} 
and M-estimation \citep{jain2014iterative}. 
The problem can be rewritten as in~\eqref{prob:main_F} 
with $\F(\Ss)=\ls(0) - \min_{\supp(\xb)\subseteq\Ss}\ls(\xb)$, 
which has SBR  
$\brus
\ge{\m{|\Us|+\s}}/{\kM{|\Us|+1}{1}}
\ge1/\kappa_{|\Us|+\s}$~\citep{elenberg2018restricted}
and  
SPR 
$\prus
\ge{\m{|\Us|+1}}/{\kM{|\Us|+\s}{\s}}
\ge1/\kappa_{|\Us|+\s}$ 
(\Cref{a_subsec:ell_0}); 
the later bound improves an 
existing result, 
$\prus\ge\mu/\nu=1/\kappa$, of 
\citep{bogunovic2018robust}. 
The evaluation of $\F(\Ss)$ involves 
solving $\min_{\supp(\xb)\subseteq\Ss}\ls(\xb)$.  
If $\ls$ is quadratic, 
we can solve it by computing a pseudo-inverse matrix.
Given a more general $\ls$,
we can use iterative methods (e.g., \citep{shalev2016accelerated})
to solve the minimization problem.

\paragraph{LP with a Cardinality Constraint}
We consider the following constrained LP that models 
optimal production planning problem \citep{bian2017guarantees}. 
Given a set of $d$ items and $\k$ production lines, 
we design a production plan so that the total profit is maximized; 
i.e., we aim to solve 
$\max_{\xb\in\Pl, \|\xb\|_0\le\k} \cb^\top\xb$, 
where 
$\cb\in\R^\d$ 
and 
$\Pl\subseteq\R^\d $ 
represent the profit of each item 
and 
a polytope specified by continuous constraints 
(e.g., upper bounds on the total quantities of materials), 
respectively.  
This problem can be reformulated as in~\eqref{prob:main_F} 
with 
$\F(\Ss)\coloneqq \max_{\xb\in\Pl} \cb_\Ss^\top\xb_\Ss$. 
As in \citep{bian2017guarantees}, 
SBR $\brus$ of $\F$ is lower bounded by some $\br_0>0$ 
for any $\Us$ and $\s$ under the {\it non-degeneracy} assumption. 
Furthermore, 
thanks to the definition of SPR, 
we have $\prus\ge1/\s$;  
although the lower bound, $1/\s$, 
can be small if $\s\approx\d$, 
this is not always the case. 
For example, 
in the guarantee of \mgreedy{} (\Cref{thm:mg}), 
$\s$ is a controllable parameter, $\bmax$; 
i.e., $\prus\ge1/\bmax$ holds.

\newcommand{\Iset}{I}
\newcommand{\ISs}{\Iset_{\Ss}}
\newcommand{\ISLs}{\Iset_{\Ss\cup\Ls}}
\newcommand{\ILs}{\Iset_{\Ls}}
\newcommand{\ILjs}{\Iset_{\Ls\cup j}}
\newcommand{\bo}{c}
\newcommand{\mset}{[m]}
\paragraph{Coverage Maximization}
Submodular functions sometimes have bounded SPR $\prus$, 
and such functions can be seen as special WM functions such that $\brus=1$ 
for any $\Us$ and $\s$. 
One such example is the coverage function. 
Let 
$V$ be a finite set and $w_v\ge0$ ($v\in V$). 
We define $\d$ groups 
$\Iset_1,\dots,\Iset_\d\subseteq V$, 
and 
we let 
$\ISs\coloneqq \bigcup_{j\in\Ss} \Iset_j$ for any $\Ss\subseteq\dset$. 
The coverage function is defined as 
$\F(\Ss)\coloneqq\sum_{v\in\ISs} w_v$, 
which is submodular and used in many scenarios 
including document summarization~\citep{lin2011class} 
and itemset mining~\citep{kumar2015fast}. 
Given $\s\in\Z_{>0}$, 
we assume that any collection of up to $\s$ 
groups 
covers every $v\in V$ at most $\bo_\s$ times; 
i.e., 
$\bo_\s\coloneqq\max_{v\in V, |\Ss|\le\s}|\{j\in\Ss \relmid{} v\in\Iset_j \}| $. 
Note that $\bo_\s\le\s$ always holds. 
In this case, 
SPR $\prus$ of $\F$ is lower bounded by $1/\bo_\s$ 
as proved in \Cref{a_subsec:coverage}.

\paragraph{Example with Unbounded Curvature}
We provide an example of LP with a cardinality constraint 
such that $\curv=1$. 
Let $\d=2$, 
$\xb=(x_1,x_2)^\top$,  
and 
$\epsilon\in(0,1)$.  
We consider a set function 
defined as
$\F(\Ss) = \max_{\supp(\xb)\subseteq\Ss}\{x_1 + \epsilon x_2 \relmid{} x_1+x_2\le 1, x_1\ge0, x_2\ge0 \}$ 
for any 
$\Ss\subseteq\dset$; 
i.e., 
$\cb=(1,\epsilon)^\top$ and $\Pl=\{\xb\in\R^2 \relmid{} x_1+x_2\le1, x_1\ge0, x_2\ge0 \}$. 
From the definitions of SBR and SPR, 
we can confirm that 
$\brus=1$ and $\prus\ge\frac{1}{1+\epsilon}$ hold 
for any $\Us$ and $\s$. 
On the other hand, we have $\curv=1$ 
since $\Fdel{\{2\}}{\{1\}}\ge(1-\curv)\F(\{1\})$
must hold for $\Fdel{\{2\}}{\{1\}} = 0$ 
and 
$\F(\{1\}) = 1$.

\section{MULTI-STAGE ALGORITHMS}\label{sec:multi}
We study multi-stage algorithms for WMM. 
Let $\Sso$ and $\xb^*$ be target solutions 
for WMM and $\ell_0$-constrained minimization, 
respectively, 
and $\ko\coloneqq|\Sso|=\|\xb^*\|_0$.  
As a warm-up, we first discuss two simple algorithms: 

\paragraph{Single-stage Algorithm}
We compute $\F(j)$ for $j\in\dset$ and let $\Ss=\argmax_{\Ss^\prime:|\Ss^\prime|\le\k}\sum_{j\in\Ss^\prime}\F(j)$. 
The algorithm requires to evaluate $\F$ only $d$ times, 
and it can find optimal solutions if $\F$ is modular. 
However, its approximation ratio becomes poor if $\F$ lacks the modularity: 
We consider a coverage maximization instance with $\d=2\k$ 
and $V=\{v_1, \dots, v_{2\k}\}$. 
Let 
$w_{v_j}=1$ and $I_{j}=\{v_j\}$ for $j=1,\dots,\k$, 
and 
let 
$w_{v_{j}}=\epsilon\ll1$ and 
$I_{j}=\{v_1,v_{j}\}$ for $j=\k+1,\dots2\k$. 
In this case, 
if $\epsilon>0$ is sufficiently small, 
the approximation ratio achieved by the single-stage algorithm is 
$\frac{1+\k\epsilon}{\k+\epsilon} = O(1/d)$.

\paragraph{Greedy Algorithm}
Starting from $\Ss=\emptyset$,
\greedy{} iteratively adds $\argmax_{j\notin\Ss}\Fdel{j}{\Ss}$ to $\Ss$
and outputs $\Ss$ after $k$ iterations.  
Given $\F$ with SBR $\br_{\Ss,\ko}$, 
\greedy{} achieves a $(1-\exp(-\br_{\Ss,\ko}))$-approximation guarantee. 
\greedy{} is, however, often costly due to the sequential evaluation of $F$, 
particularly when the evaluation of $\F$ involves solving optimization problems. 
For example, 
in the case of $\ell_0$-constrained minimization, 
\greedy{} solves convex minimization problems $\Thetarm(dk)$ times. 

Namely, while the single-stage algorithm 
can efficiently find optimal solutions if $\F$ is modular, 
\greedy{} can achieve better guarantees for non-modular $\F$ 
at the cost of more computational effort. 
In the case of WMM, 
since $\F$ belongs to a class that is close to modular functions, 
we can expect that an intermediate of the above two algorithms works well.
The multi-stage approach provides such an intermediate. 
As in \Cref{alg:multi}, 
we perform $\mm$ ($\le\k$) iterations to obtain a solution. 
In each $i$-th iteration, 
we choose a subset $\Bs_i\subseteq\dset$ 
of size at most $b_i$ so that it 
maximizes a {\it surrogate function}, 
$\Gs$, 
where $\Ss$ is the current solution. 
To obtain fast multi-stage algorithms, 
$\Gs$ should be evaluated and maximized efficiently. 
Below we design $\Gs$ for 
\mgreedy{} and \momp{}, 
and we present their theoretical guarantees.  
We then experimentally evaluate the multi-stage algorithms.

\subsection{Theoretical Guarantees}\label{subsec:multi_guarantees}
\begin{algorithm}[tb]
	\caption{Multi-stage algorithm}
	\label{alg:multi}
	\begin{algorithmic}[1]
		\State $\Us\gets \dset$, $\Ss\gets\emptyset$
		\For{$i=1,\dots,\mm$}
		\State $\Bs_i\gets\argmax_{\Bs\subseteq \Us:|\Bs|\le b_i} 
		\Gs(\Bs)$
		\State $\Ss\gets\Ss\cup\Bs_i$
		\State $\Us\gets \Us\bs\Bs_i$
		\EndFor
		\State \Return $\Ss$
	\end{algorithmic}
\end{algorithm}

Let $\Ss_i=\Bs_1\cup\cdots\cup\Bs_i$ 
for $i\in[\mm]$ and $\Ss_0=\emptyset$.
We first present a guarantee of \mgreedy{} 
for WMM. 
We then focus on $\ell_0$-constrained minimization 
and prove a guarantee of \momp{}. 
As detailed below, 
our results generalize and improve some existing results, 
which emphasizes that to utilize the weak modularity 
is effective for obtaining strong theoretical results. 
The proofs of the theorems are presented in \Cref{a_subsec:multi_guarantee}.

\subsubsection{Multi-Greedy}\label{subsubsec:mg}
\mgreedy{} uses 
$
\Gs(\Bs) = \sum_{j\in\Bs} \Fdel{j}{\Ss}
$ as a surrogate function. 
As a result, \mgreedy{} evaluates $\F$ $\Thetarm(dm)$ times. 
We can show that \mgreedy{} enjoys the following approximation guarantee: 
\begin{thm}\label{thm:mg}
	Let $\bmax$ be an integer satisfying $1\le\bmax\le\ko$. 
	Set $b_1,\dots,b_m$ so as to satisfy  
	$b_i\in[\bmax]$ for $i\in[\mm]$ and $\sum_{i\in[\mm]}b_i=\k$.  
	If 
	$\Ss$ is the solution 
	obtained with \mgreedy{} 
	and 
	$\F$ is $(\br_{\Ss,\ko}, \pr_{\Ss,\bmax})$-WM, 
	we have
	\begin{align}
	\F(\Ss)
	\ge 
	\left(1- 
	\prod_{i=1}^\mm
	\left(1 -
	\br_{\Ss_{i-1},\ko}
	\pr_{\Ss_{i-1},b_i}
	\frac{b_i}{\ko}
	\right)  \right)
	\F(\Sso)
	\ge 
	\left(1-\exp\left( -
	\br_{\Ss,\ko}
	\pr_{\Ss,\bmax}
	\frac{\k}{\ko}
	\right)  \right)
	\F(\Sso). 
	\end{align}
\end{thm}
Note that, 
if we set $\bmax=1$, 
this result recovers the 
$(1-e^{-\br_{\Ss,\k}})$-approximation 
of 
\greedy{}~\citep{das2018approximate} 
since $\pr_{\Ss,1}=1$. 
\mgreedy{} with $m=1$ is sometimes called the oblivious algorithm 
in the field of $\ell_0$-constrained minimization. 
\citet{elenberg2018restricted} proved that its approximation ratio is at least  $\max\left\{\frac{1}{\k}\kappa_\k^{-1}, \frac{3}{4}\kappa_{k}^{-2}, \kappa_{k}^{-3}\right\}$. 
Note that \Cref{thm:mg} improves this result since, 
if $b_1=\k=\ko$, 
the approximation ratio becomes 
$\max\{\frac{1}{\k}\kappa_{k}^{-1}, \kappa_{k}^{-2}  \}$ 
thanks to 
the lower bounds of SBR and SPR (\Cref{sec:applications}) 
and 
$\pr_{\emptyset,\k}\ge1/\k$. 
More generally, given $m\ge1$, \mgreedy{} 
achieves a $1-\exp(\kappa^{-1}_{2\k}\kappa^{-1}_{\k+\bmax})$-approximation guarantee. 
Below we show that a stronger guarantee for $\ell_0$-constrained minimization 
can be obtained by using \momp{}.

\subsubsection{Multi-OMP}
We then focus on $\ell_0$-constrained minimization; 
i.e., 
we assume $\F(\Ss)=\ls(0)-\min_{\supp(\xb)\subseteq\Ss}\ls(\xb)$ 
($\forall \Ss\subseteq\dset$).
Let $\bb{\Ss}\coloneqq\argmin_{\supp(\xb^\prime)\subseteq\Ss}\ls(\xb^\prime)$. 
\momp{} uses 
$
\Gs(\Bs) 
=
\sum_{j\in\Bs} |\nabla\ls(\bb{\Ss})_j|^2
$
as a surrogate function; 
thus, it requires 
to compute the gradient and 
to solve convex minimization problems $m$ times. 
To prove the guarantee of \momp{}, 
we use the following lemma, 
which, 
roughly speaking, connects the decrease in $\ls$ to the increase in $\F$.

\begin{lem}\label{lem:fts}
	For
	any disjoint $\As,\Bs\subseteq\dset$, 
	if $\ls(\cdot)$ is 
	$\m{|\As\cup\Bs|}$-RSC and 
	$\M{|\Bs|,|\Bs\bs\As|}$-RSM, 
	we have
	\begin{align}
	\frac{\|\nabla \ls(\bb{\As})_{\Bs} \|^2_2}{2\M{|\Bs|,|\Bs\bs\As|}} 
	\le
	\Fdel{\Bs}{\As}
	\le
	\frac{\|\nabla \ls(\bb{\As})_{\Bs} \|^2_2}{2\m{|\As\cup\Bs|}}.
	\end{align}
\end{lem}
A special case of the lemma is implicitly used in \citep{elenberg2018restricted}. 
In \Cref{a_subsec:ell_0}, 
we provide a slightly stronger version of the lemma, 
which we use for proving the guarantee of \momp{}. 
By using the lemma, we can employ the technique used 
when proving \Cref{thm:mg} and obtain the following result:
\begin{thm}\label{thm:mo}
	Set $b_1,\dots,b_\mm$ as in \Cref{thm:mg}. 
	If $\ls$ is 
	$\m{\k+\ko}$-RSC and $\M{\k,\bmax}$-RSM, 
	then \momp{} outputs solution $\Ss$ such that  
	$\xb=\argmin_{\supp(\xb^\prime)\subseteq\Ss}\ls(\xb^\prime)$ satisfies 
	\begin{align}
	\ls(\xb)\le{}&
	\ls(\xb^*)
	+
	\prod_{i=1}^\mm
	\left(1 -
	\frac{\m{|\Ss_{i-1}\cup\Sso|}}{\M{|\Ss_{i}|,|\Bs_i|}}
	\frac{b_i}{\ko}
	\right)
	(\ls(0) - \ls(\xb^*))
	\\
	\le{}&
	\ls(\xb^*)
	+
	\exp\left( -
	\frac{\m{\k+\ko}}{\M{\k,\bmax}}
	\frac{\k}{\ko}
	\right)
	(\ls(0) - \ls(\xb^*))
	\\
	\le{}&
	\ls(\xb^*)
	+
	\exp\left( -
	\frac{1}{\kappa_{\k+\ko}}
	\frac{\k}{\ko}
	\right)
	(\ls(0) - \ls(\xb^*)).
	\end{align}
\end{thm}
Note that, if $\k=\ko$, 
\Cref{thm:mo} gives a ($1-\exp(\kappa^{-1}_{2\k})$)-approximation guarantee, 
which improves the aforementioned guarantee of \mgreedy{}. 
Interestingly, 
the approximation ratio matches those of \omp{} 
and \greedy{} \citep{elenberg2018restricted}. 
Namely, 
the use of the multi-stage approach does not degrade 
the theoretical guarantee. 
If we let $b_1=\k=\ko$, 
\momp{} with only one round achieves 
a $\kappa^{-1}_{2\k}$-approximation guarantee, 
which improves the existing 
$\left(1-\exp\left(-{\kappa_{2\k}^{-4}}
\right)\right)$-approximation 
guarantee with $O(\log(\d))$ rounds, 
recently proved by \citet{qian2019fast}, 
in terms of both the approximation ratio and the computation complexity. 


\subsection{Experiments }\label{subsec:experiments}
\begin{figure*}[tb]
	\vspace{0pt}
	\centering
	\includegraphics[width=1.0\linewidth]{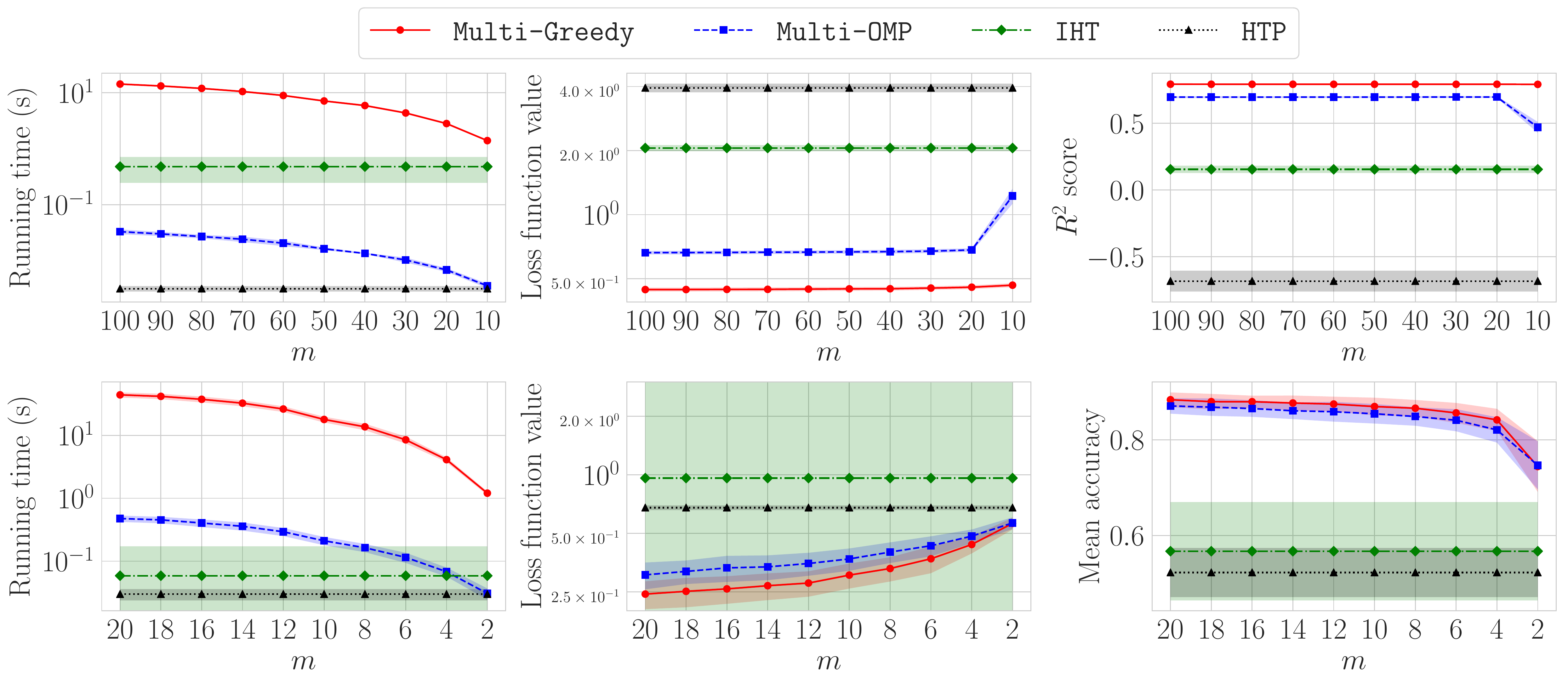}
	\caption{
		Results  of $\ell_0$-constrained Minimization with Various $m$ Values. 
		Top (bottom) figures present those of 
		regression (classification) instances. 
		Running times and loss function values are 
		shown with semi-log plots.  
		Each curve and error band indicate 
		the average and standard deviation, 
		respectively, 
		calculated over 
		$100$ instances. 
	}
	\label{fig:real_plot}
\end{figure*}
We evaluate 
the multi-stage algorithms 
via experiments with two kinds instances:  
$\ell_0$-constrained minimization 
and LP with a cardinality constraint. 
We use Python3 to implement the algorithms, 
and we conduct experiments 
on a 64-bit macOS 
machine
with 3.3GHz Intel Core i7 CPUs and 16 GB RAM. 
All the algorithms considered below
can be accelerated via randomization \citep{li2016stochastic,khanna2017scalable}, 
but 
to simplify the comparisons 
we here do not employ such techniques.

\subsubsection{$\ell_0$-constrained Minimization}\label{subsubsec:real}

We use two instances with the real-world dataset available at PMLB \citep{olson2017pmlb}. 
The first is a sparse regression instance 
with the square loss, 
$\ls(\xb)=\frac{1}{2n}\|\yb - \A\xb\|_2^2$, 
where $\A\in\R^{n\times\d}$ and $\yb\in\R^n$ 
are obtained from ``satellite\_image'' dataset. 
We use the 1st and 2nd order polynomial features; 
as a result, 
we have 
$d=666$ features and 
a sample of size $N=6435$.  
We 
set $\k=100$. 
The second is a sparse classification instance. 
We use the regularized logistic loss, 
$\ls(\xb)=\frac{1}{n}\sum_{i\in[n]}\log({1+\exp(-\yb_i (\A\xb)_i})) + \frac{\lambda}{2}\|\xb\|_2^2$, 
where $\A\in\R^{n\times\d}$ and $\yb\in\R^n$ 
are obtained from ``hill\_valley\_with\_noise'' dataset. 
The dataset has 
$d=100$ features and 
a sample of size $N=1212$. 
We let $\lambda=0.01$ and $\k=20$.  
For each instance, 
we randomly split the sample into training and test data 
of sizes $\ceil{N/2}$ and $\floor{N/2}$, respectively; 
we thus create $100$ random instances. 
We consider multi-stage algorithms with various numbers of iterations, 
$m=\k, 0.9\k, \dots, 0.1\k$ 
($m=\k$ corresponds to standard \greedy{}/\omp{}).  
We set  
$b_1,\dots, b_{\k-m\floor{\k/m}}$ at $\ceil{k/m}$ 
and 
the rest at $\floor{k/m}$. 
We use two baselines based on the projected gradient method: 
iterative hard thresholding (\iht{}) \citep{jain2014iterative} 
and hard thresholding pursuit (\htp{}) \citep{yuan2018gradient}. 
We continue their iterations until
the decrease in $\ls(\cdot)$ value 
becomes smaller than $10^{-5}$. 
We evaluate the algorithms with 
running times, loss function values, 
$R^2$ scores (for regression), 
and 
mean accuracy (for classification); 
the last two are defined by the corresponding 
scikit-learn score functions.

Figure~\ref{fig:real_plot} summarizes the results. 
We see that the multi-stage algorithms speed up as $m$ decreases; 
in particular, 
\momp{} becomes as fast as \htp{}. 
In the regression instances, 
multi-stage algorithms achieve better  
loss function values and $R^2$ scores than the baselines. 
Other than for \momp{} with $m=10$, 
the decrease in $m$ has negligible effects on 
loss function values and $R^2$ scores. 
In the classification instances, 
the loss function values of the multi-stage algorithms increase as $m$ decreases, 
but they are smaller on average than those of \iht{} and \htp{}. 
The multi-stage algorithms also achieve 
better mean accuracy than the baselines. 
To conclude, 
by using the multi-stage approach, 
\greedy{} and \omp{} can become faster while outperforming the baselines. 
When addressing large-scale instances in practice, it would be effective to try multi-stage algorithms 
with a small $m$ and increase it until an acceptable solution is obtained. 

As regards solution quality, 
the performance gap between the 
greedy-style algorithms (\mgreedy{} and \momp{}) 
and the baselines (\iht{} and \htp{}) 
can partially be explained in terms of 
the restricted condition number. 
For example, \iht{} requires $\k\ge \Omegarm(\kappa_{2\k+\k^*}^2)$ 
to achieve $\epsilon$-errors~\citep{jain2014iterative}, 
while \momp{} requires $\k\ge \Omegarm(\kappa_{\k+\k^*})$ 
as implied in \Cref{thm:mo}. 
This suggests that 
greedy-style algorithms can be more resistant 
to being ill-conditioned (or a large restricted condition number), 
which is often the case with 
real-world instances; 
hence the better performance of the greedy-style algorithms. 
\Cref{a_subsec:experiments_synthetic} 
presents further experiments with 
well- and ill-conditioned instances. 

\subsection{LP with a Cardinality Constraint}
We use synthetic optimal production planning instances. 
We let 
$\Pl=\{\xb\in\R^\d\relmid{}\A\xb\le\bvec, \zeros\le\xb\le\ones \}$.  
Each entry of $\A\in\R^{m\times \d}$ 
and $\cb\in\R^{\d}$ 
is drawn from the uniform distribution on $[0,1]$. 
We set 
$\d=50$, $m=100$, and $\bvec=0.5\k\times\ones$. 
We consider various sparsities $\k=5,10,\dots,50$; 
for each $\k$, we randomly generate $100$ instances as above. 
We consider \mgreedy{} with $m=2$ and $m=5$,  
denoted by \mgreedy{-}$2$ and \mgreedy{-}$5$, respectively. 
As baselines, 
we employ \greedy{} and \random, 
which chooses $\k$ elements from $\dset$ uniformly at random. 

\begin{figure}
	\centering
	\includegraphics[width=1.0\linewidth]{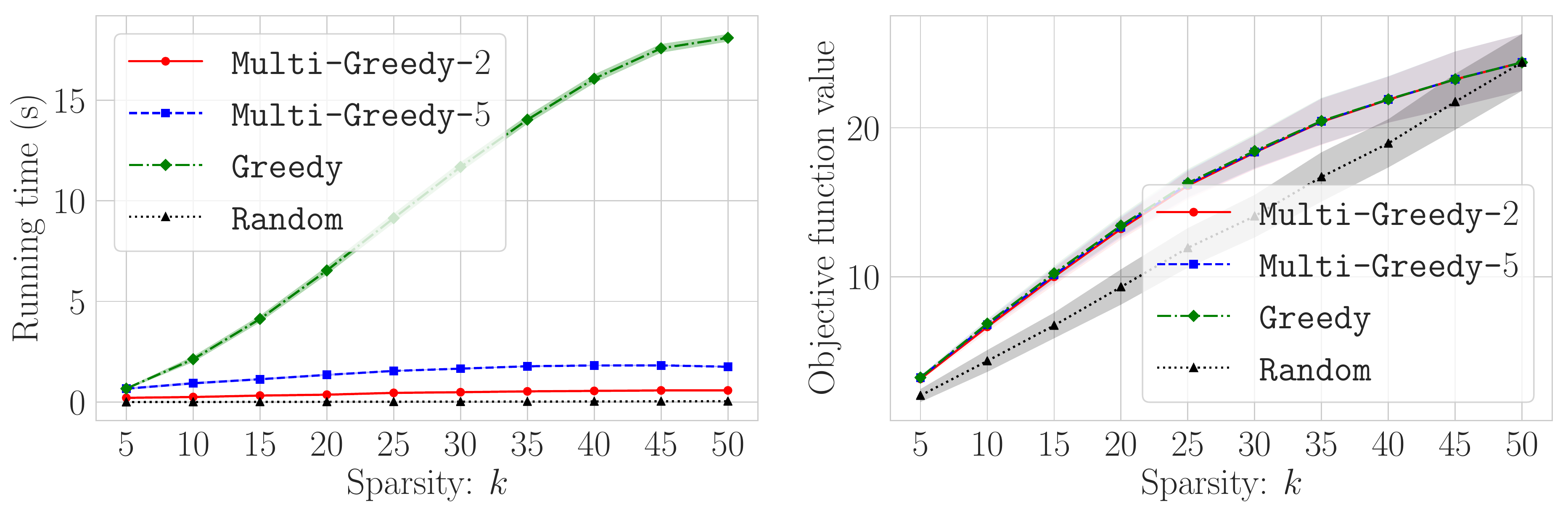}
	\caption{
		Results of LP with a Cardinality Constraint. 
		Each curve (error band) indicates 
		the average (standard deviation) 
		calculated over 
		$100$ instances. 
	}
	\label{fig:lp_plot}
\end{figure}
\Cref{fig:lp_plot} shows the results. 
We see that \mgreedy{} 
algorithms run far faster than 
\greedy{}, 
and they achieve almost the same 
objective values as those of \greedy{}. 
Namely, 
for optimal production planning instances, 
the multi-stage strategy can accelerate 
\greedy{} at a very slight sacrifice of solution quality.

\section{THEORETICAL PROPERTIES}\label{sec:theoretical}
We study theoretical properties of WMM: 
In \Cref{sec:fpt}
we show that WMM is  fixed-parameter tractable (FPT)  under certain conditions, 
and 
in \Cref{sec:hard}
we 
prove that no polynomial-time algorithms can improve 
the $(1-e^{-\br_{\Ss,\k}})$-approximation guarantee 
even if SBR and SPR are bounded by some constants.  

\subsection{Fixed-parameter Tractability}\label{sec:fpt}
Here we discuss the computation cost of solving WMM almost optimally.
If we are to find an optimal solution for WMM, 
a naive approach is exhaustive search; 
i.e., we examine $\F(\Ss)$ 
for all $\Ss\subseteq\dset$ of size $\k$. 
This, however, incurs $\Omegarm(\d^\k)$ computation cost, 
which becomes too large as the instance size, $\d$, increases. 
Taking this into account, 
the following question arises: 
Can we solve WMM (almost) optimally 
without requiring an $\Omegarm(\d^\k)$ computation cost? 
To answer this, 
we use the {parameterized complexity} framework \citep{cygan2015parametrized}. 
We regard a part of the input as a fixed parameter(s),  
which 
is denoted by ${\bf p}$ 
and 
does not include the instance size, $\d$. 
An algorithm is said to be FPT if it runs in $g({\bf p})\times\poly(\d)$ time,
where $g$ is a computable function of ${\bf p}$.
Note that, if $k$ is a fixed parameter, 
algorithms that require $\Omegarm(\d^k)$ time, 
including exhaustive search, are not FPT.
Here, 
regarding $\k$ as a part of the fixed parameters, 
we show that $\epsilon$-error solutions for WMM can be computed 
with a randomized FPT algorithm (\Cref{alg:fpt}), 
which was originally developed by \citet{skowron2017fpt} 
for a special case of 
monotone submodular maximization. 
\Cref{alg:fpt} performs
$\singlerun$,
a randomized variant of \greedy{},
$T$ times
and returns the best solution. 
We can show that it enjoys the following guarantee for WMM:   
\begin{algorithm}[tb]
	{
		\caption{Randomized FPT algorithm}
		\label{alg:fpt}
		\begin{algorithmic}[1]
			\State 
			\begin{tabular}{@{}p{20cm}}
				Execute \singlerun~$T$ times
				and return the best solution.
			\end{tabular}
			\Function{\singlerun}{}
			\State $\Ss_{0}\gets\emptyset$
			\For{$i=1,\dots,k$}
			\State
			\begin{tabular}{@{}p{20cm}}
				Choose $j\in\dset\bs\Ss_{i-1}$ randomly with probability 
				$\propto\Fdel{j}{\Ss_{i-1}}$.
			\end{tabular}
			\State $\Ss_{i}\gets\Ss_{i-1}\cup\{j\}$
			\EndFor
			\State \return~$\Ss_k$
			\EndFunction
		\end{algorithmic}
	}
\end{algorithm}
\begin{thm}\label{thm:fpt}
	Assume $\F$ to be ($\brkk, \pr_{\k,\d}$)-WM. 
	Let $\Sso$ be an optimal solution for problem~\eqref{prob:main_F} 
	and 
	$\tilde{\F}\coloneqq\F(\dset)-\F(\Sso)$.   
	For any $\epsilon>0$,
	if 
	\[
	T\ge\Tfpt,
	\] 
	then \Cref{alg:fpt} returns solution $\Ss$ satisfying
	$\F(\Ss)\ge\F(\Sso)-\epsilon$ with a probability of at least $1-\delta$.
\end{thm}
The key to proving \Cref{thm:fpt} is the fact that 
the probability of choosing $j\in\Sso$ in each iteration 
can be lower bounded thanks to 
the weak modularity. 
We present the proof in \Cref{a_sec:fpt} for details.

Since  $\F$ can be evaluated in $\poly(d)$ time 
as assumed in \Cref{subsec:background}, 
\Cref{alg:fpt} is FPT 
if we regard ${\mathbf{p}}\coloneqq (k,\brkk,\pr_{\k,\d}, \tilde{\F}, 
\epsilon,\delta)$ as fixed parameters. 
Note that, 
since $\tilde{\F}\le\F(\dset)$, 
a sufficiently large $T$ can be computed once 
we obtain lower bounds of SBR and SPR, 
which are available for various applications as in \Cref{sec:applications}. 

While \Cref{alg:fpt} is not so practical, 
\Cref{thm:fpt} is beneficial for studying the tractability of WMM instances. 
In particular, we can obtain an interesting corollary 
related to $\ell_0$-constrained minimization from the theorem. 
Let $\xb^*\coloneqq\argmin_{\|\xb\|_0\le\k}\ls(\xb)$ 
be an optimal solution. 
As shown by~\citet{shalev2010trading}, 
\greedy{} can find $\xb$ such that $\ls(\xb)\le\ls(\xb^*) + \epsilon$
if $\xb$ is allowed to have $\Omegarm(\kappa\log\epsilon^{-1})$ non-zeros; 
i.e., there is a trade-off between sparsity $\|\xb\|_0$ and accuracy $\epsilon$. 
In practice, however, 
$\xb$ is not always allowed to have sufficiently many non-zeros. 
For instance, 
when performing feature selection for medical analysis, 
the number of features used for predicting a patient's status 
is limited since to use many features 
requires the patient to undergo many medical tests, 
which is a considerable burden. 
Hence, 
to reveal whether we can solve $\ell_0$-constrained minimization 
almost optimally with limited sparsity $\k\ge\|\xb\|_0$ 
is an important research subject. 
The following corollary, 
which is 
obtained from \Cref{thm:fpt} and 
the lower bounds of SBR/SPR (\Cref{sec:applications}), 
implies that it is possible at the cost of FPT computation time; 
i.e, 
there is a time--accuracy trade-off.

\begin{cor}\label{cor:fpt} 
	Let 
	$\F(\Ss)=\ls(0) - \min_{\supp(\xb)\subseteq\Ss}\ls(\xb)$ 
	for any $\Ss\subseteq\dset$ 
	and assume $\ls$ to be 
	$\m{2k}$-RSC, 
	$\m{k+1}$-RSC, 
	$\M{k+1,1}$-RSM, 
	and 
	$\M{d}$-RSM. 
	Let  
	${\tilde\ls}\coloneqq \ls(\xb^*) - \min_{\xb\in\R^\dset}\ls(\xb)$. 
	If Algorithm~\ref{alg:fpt} runs with 
	\[
	T\ge
	\left\lceil\left(
	\frac{{\M{k+1,1}}}{{\m{2k}}}
	\cdot
	\frac{{\M{d}}}{{\m{k+1}}}
	\cdot
	\frac{{\tilde\ls}+\epsilon}{\epsilon}
	\right)^k
	\log\delta^{-1}
	\right\rceil
	\]	
	and outputs $\Ss$, 
	then 
	$\xb=\argmin_{\supp(\xb^\prime)\subseteq\Ss}\ls(\xb^\prime)$ 
	satisfies 
	$\ls(\xb)\le\ls(\xb^*)+\epsilon$ 
	with a probability of at least $1-\delta$.
\end{cor}
Namely, if we take ${\mathbf{p}}\coloneqq (k,\m{2k}, \m{k+1}, \M{k+1,1}, \M{d}, {\tilde \ls}, \epsilon,\delta)$ 
to be fixed parameters, 
$\epsilon$-error solutions 
can be computed by using \Cref{alg:fpt} with a high probability. 
Note that, 
unlike the above guarantee of \greedy{}, 
\Cref{cor:fpt} does not require $\|\xb\|_0$ to be sufficiently large. 

\subsection{Hardness Result}\label{sec:hard} 
We here prove the following hardness of improving the ($1-e^{-\br_{\Ss,\k}}$)-approximation guarantee for WMM: 
\begin{thm}\label{thm:hard}
	Even if $\F$ has SBR $\brkk=1$ and SPR $\prkk\ge1/2-o(1)$, 
	no algorithms that evaluate $\F$ only on polynomially many subsets can 
	achieve an approximation guarantee that exceeds 
	$1-e^{-1}=1-e^{-\brkk}$ for problem~\eqref{prob:main_F} in general.  
\end{thm}
Note that the significance of \Cref{thm:hard} comes from 
SPR $\prkk$ that can be bounded by a universal constant: 
When curvature $\curv$, 
which satisfies $\prkk\ge1-\curv$, 
is upper bounded by a universal constant smaller than $1$, 
then a strictly improved approximation ratio, 
$\frac{1}{\curv}(1-e^{-\curv\br_{\Ss,\k}})$,  
can be obtained thanks to \citep{bian2017guarantees}. 
Namely, \Cref{thm:hard} reveals a non-trivial theoretical gap 
between SPR $\prkk$ and curvature $\curv$. 

\begin{proof}[Proof sketch]
	We make a WM function that is hard to maximize approximately. 
	As with the proof of \citep{nemhauser1978best},  
	given unknown subset $\Ms$ of size $\k$, 
	we show that to achieve an approximation guarantee 
	that exceeds $1-e^{-\brkk}$ 
	is at least as hard as to find $\Ss$ such that $|\Ss\cap\Ms|>r$ and $|\Ss|\le \pk\coloneqq 2\k-r+1$, 
	where $r>0$ is any fixed integer; 
	note that this cannot be solved via polynomially many queries. 
	To show it, 
	we use $\F$ that satisfies the following conditions: 
	$\F(\Ss)$ value depends on $|\Ss|$ and $|\Ss\cap\Ms|$ for any $\Ss\subseteq\dset$ 
	and only on $|\Ss|$ if $|\Ss\cap\Ms|\le r$ or $|\Ss|>\pk$, 
	which, roughly speaking, 
	means that the information about $\F$ values is useless.  
	By using such function $\F$, we can obtain the hardness result. 
	The main difficulty remained in the above proof is to 
	show that $\F$ is WM. 
	In particular,  
	obtaining $\prkk\ge1/2-o(1)$ is the most challenging part.  
	To prove this, we first rewrite SPR as $\prkk=\min_{\Ls,\Ss\subseteq\dset}\left\{\frac{\Fdel{\Ss}{\Ls}}{\sum_{j\in \Ss}\Fdel{j}{\Ls}} \relmid{\Big} \Ls\cap \Ss=\emptyset, |\Ls|\le\k, |\Ss|\le\k \right\}$, 
	where we regard $0/0=1$. 
	Then, by carefully designing $\F$ and using the fact that $\F(\Ss)$ depends only on $|\Ss|$ and $|\Ss\cap\Ms|$, 
	we express $\prkk$ as the optimal value of 
	\begin{align}
	\minimize_{x,y,z\in\R}
	&
	\quad 
	\frac{z - (z - x)\left(1 - 1/\k \right)^{k(y-x)} }
	{x(1-z)+yz} 
	\\
	\subto
	&
	\quad 
	0\le x \le y \le \frac{\k}{\k-r+1}, 
	\quad 
	0 \le x \le z \le 1,
	\quad
	\text{and}
	\quad   
	y-z \le \frac{\d}{\k-r+1} - 1, 
	\end{align}
	which can be lower bounded by $\frac{1}{2} - \frac{1}{2}\cdot\frac{r-1}{2k-r+1}$. 
	By letting $\k$ increase with $\d$ and 
	setting $\d$ at a sufficiently large value, we complete the proof 
	(see, \Cref{a_subsec:hardness_proof} for the full proof).  
\end{proof}

Given solution $\Ss$ of \greedy{}, 
we always have $\br_{\Ss,\k}\ge\brkk$. 
Therefore, \Cref{thm:hard} 
implies that, 
even if $\prkk$ ($\le\pr_{\Ss,\k}$) is lower bounded 
by a value that can be arbitrarily close to $1/2$, 
no polynomial-time algorithms 
can improve the $(1-e^{-\br_{\Ss,\k}})$-approximation guarantee in general. 

We remark that it may be possible to improve 
the approximation ratio for some easier subclasses of WMM; 
for example, if $\br_{\d}$ ($\le\brkk$) and $\pr_{\d}$ ($\le\prkk$) 
are bounded, 
we may be able to obtain 
a better ratio than $1-e^{-\br_{\d}}$ 
by using $\pr_{\d}$. 
We discuss this topic in \Cref{a_subsec:discussion}. 
We also remark that \Cref{thm:hard} does not contradict 
the FPT result, \Cref{thm:fpt}, for the following reason: 
\Cref{thm:hard} is proved 
by using sparsity $k$ that increases with $d$, 
and we cannot regard such a $k$ as a fixed parameter.

\section{CONCLUSION}\label{sec:conclusion}
We studied WMM, a class of non-submodular maximization 
that can model various practical problems. 
We proved guarantees of multi-stage algorithms, 
which generalize and improve some existing results,  
and confirmed their effectiveness experimentally. 
We then proved 
the fixed-parameter tractability of WMM, 
which yields the time--accuracy trade-off for $\ell_0$-constrained minimization as a byproduct, and 
the hardness of improving the ($1-e^{-\br_{\Ss,\k}}$)-approximation guarantee.
Recent studies
\citep{khanna2017scalable,qian2018approximation} 
have provided various techniques for accelerating 
greedy algorithms, 
and 
greedy-style methods for many different 
settings have also been studied~\citep{bogunovic2018robust,fujii2018fast}. 
It will be interesting future work 
to study how to 
incorporate the multi-stage approach 
into those methods for further acceleration. 

{
	\bibliography{mybib}
	\bibliographystyle{apalike}
}

\appendix
\setcounter{equation}{0}
\setcounter{thm}{0}
\setcounter{prop}{0}
\setcounter{lem}{0}
\setcounter{algorithm}{0}
\renewcommand{\theequation}{A.\arabic{equation}}

\theoremstyle{theorem}
\newtheorem{alem}{Lemma}
\setcounter{alem}{0}
\renewcommand{\thealem}{A.\arabic{alem}}
\newtheorem{aprop}{Proposition}
\setcounter{aprop}{0}
\renewcommand{\theaprop}{A.\arabic{aprop}}

\clearpage

\onecolumn
\begin{center}
	{\fontsize{18pt}{0pt}\selectfont \bf Appendices}
\end{center}

In \Cref{a_sec:applications}, 
we derive the lower bounds of SBR and SPR for each application 
presented in \Cref{sec:applications}, 
and we also provide example instances where $\curv=1$ holds 
even if SBR and SPR are bounded. 
In \Cref{a_sec:multi}, we prove the guarantees of the multi-stage algorithms, 
and  
we also present additional experimental results with well- and ill-conditioned 
synthetic $\ell_0$-constrained minimization instances. 
In \Cref{a_sec:fpt}, we prove the guarantee of the FPT algorithm. 
In \Cref{a_sec:hardness}, we present the proof of the hardness result.

\section{Applications}\label{a_sec:applications}
We show that SBR and SPR are lower bounded 
for each application presented in~\Cref{sec:applications}. 
We also present example instances 
such that $\curv=1$ holds even if SBR and SPR are lower bounded; 
regarding $\ell_0$-constrained minimization, 
we show that {\it inverse curvature} $\icurv\in[0,1]$~\citep{bogunovic2018robust} can also become equal to $1$. 
Note that curvature $\curv$ and inverse curvature $\icurv$ of 
$\F$ are defined as the smallest scalars that satisfy
\begin{align}
{\Fdel{j}{\Ss\bs\{j\}\cup\Ms}}
\ge
(1-\curv)
{\Fdel{j}{\Ss\bs\{j\}}} 
& & 
\text{and} 
& & 
{\Fdel{j}{\Ss\bs\{j\}}}
\ge
(1-\icurv)
{\Fdel{j}{\Ss\bs\{j\}\cup\Ms}},  
\end{align} 
respectively, 
for any 
$\Ss,\Ms\subseteq\dset$ and $j\in\Ss\bs\Ms$. 
Function $\F$ is submodular (supermodular) 
iff
$\icurv=0$ ($\curv=0$).  
As shown in~\citep{bogunovic2018robust}, 
for any $\Us\subseteq\dset$ and $\s\in\Z_{>0}$, 
we have 
\begin{align}
\brus\ge 1-\icurv 
& & 
\text{and}
& &
\prus\ge 1-\curv.
\end{align}
Namely, 
while the bounded curvature (inverse curvature) implies bounded SPR (SBR), 
the opposite is not always true.

\subsection{$\ell_0$-constrained Minimization}\label{a_subsec:ell_0}
\paragraph{Lower Bounds of SBR and SPR}
We first introduce some definitions required in the following discussion. 
Given $\Omegarm\subseteq\R^\dset\times\R^\dset$,
we say $\ls$ is $\m{\Omegarm}$-RSC 
and $\M{\Omegarm}$-RSM 
if it satisfies 
\begin{equation}
\frac{\m{\Omegarm}}{2}\|\yb-\xb\|_2^2
\le \ls(\yb) - \ls(\xb)
-
\iprod{\nabla l(\xb)}{\yb-\xb}
\le
\frac{\M{\Omegarm}}{2}\|\yb-\xb\|_2^2 
\end{equation}
for all $(\xb,\yb)\in\Omegarm$. 
For convenience,
we define $\f(\xb)\coloneqq \ls(0)-\ls(\xb)$. 
Note that we have 
\[
\F(\Ss)
=\ls(0) - \min_{\supp(\xb)\subseteq\Ss}\ls(\xb)
=\max_{\supp(\xb)\subseteq\Ss}\f(\xb)
\]
for any $\Ss\subseteq\dset$.
If $\ls$ is
$\m{\Omegarm}$-RSC and $\M{\Omegarm}$-RSM,
then $\f$ is 
$\m{\Omegarm}$-restricted strong concave ($\m{\Omegarm}$-RSC)
and $\M{\Omegarm}$-restricted smooth ($\M{\Omegarm}$-RSM) as follows:
\begin{equation}
-\frac{\m{\Omegarm}}{2}\|\yb-\xb\|_2^2
\ge
\f(\yb) - \f(\xb) -
\iprod{\nabla \f(\xb)}{\yb-\xb}
\ge
-\frac{\M{\Omegarm}}{2}\|\yb-\xb\|_2^2 \label{a_def:rscrsm}
\end{equation}
for any $(\xb,\yb)\in\Omegarm$.
We employ the following definitions for convenience:
\begin{itemize}
	\item
	If~\eqref{a_def:rscrsm} holds with
	\[
	\Omegarm=\Omegarm_{\s_1, \s_2}\coloneqq
	\{(\xb,\yb) \relmid{} \|\xb\|_0\le \s_1, \|\yb\|_0\le \s_1, 
	\text{ and }
	\|\xb-\yb\|_0\le \s_2 \},
	\]
	we say $\f$ is
	$\km{\s_1}{\s_2}$-RSC and $\kM{\s_1}{\s_2}$-RSM.
	For simplicity, we define $\m{\s}\coloneqq\km{\s}{\s}$ and $\M{\s}\coloneqq\kM{\s}{\s}$.
	\item
	Given $\As,\Bs\subseteq \dset$,
	if~\eqref{a_def:rscrsm} holds with
	\[
	\Omegarm=\Omegarm_{\As,\Bs}\coloneqq
	\{(\xb,\yb) \relmid{} \supp(\xb)\subseteq \As, \supp(\yb)\subseteq \Bs \},
	\]
	we say $\f$ is
	$\m{\As,\Bs}$-RSC and $\M{\As,\Bs}$-RSM.
	\item
	Given $\As\subseteq\Bs\subseteq \dset$,
	if~\eqref{a_def:rscrsm} holds with
	\[
	\Omegarm=\tOmegarm_{\As,\Bs}\coloneqq
	\{(\xb,\yb) \relmid{} \supp(\xb)\subseteq \As, \supp(\yb)\subseteq \Bs, 
	\text{ and } 
	\supp(\yb-\xb)\subseteq\Bs\bs\As \},
	\]
	we say $\f$ is $\tm{\As,\Bs}$-RSC and $\tM{\As,\Bs}$-RSM.
\end{itemize}
Given any $\Omegarm^\prime$ and $\Omegarm$ 
satisfying $\Omegarm^\prime\subseteq\Omegarm$,
we can set $\m{\Omegarm^\prime}$ and $\M{\Omegarm^\prime}$
so that we have
$\m{\Omegarm^\prime}\ge\m{\Omegarm}$ and
$\M{\Omegarm^\prime}\le\M{\Omegarm}$, respectively.
In particular,
we often use the following inequalities:
\begin{itemize}
	\item
	For any $0\le \s_1^\prime\le \s_1$ and $0\le \s_2^\prime\le \s_2$ we have
	$\km{\s_1}{\s_2}\le\km{\s_1^\prime}{\s_2^\prime}$ and
	$\kM{\s_1}{\s_2}\ge\kM{\s_1^\prime}{\s_2^\prime}$.
	\item
	For any $\As,\Bs\subseteq\dset$, we have
	$\m{|\As\cup\Bs|}\le\m{\As,\Bs}$
	and
	$\M{|\As\cup\Bs|}\ge\M{\As,\Bs}$.
	\item
	For any $\As\subseteq\Bs\subseteq\dset$,
	we have,
	$\km{|\Bs|}{|\Bs\bs\As|}\le\tm{\As,\Bs}$ and
	$\kM{|\Bs|}{|\Bs\bs\As|}\ge\tM{\As,\Bs}$.
\end{itemize}
%
The following lemma  
is the key to obtaining the lower bounds of SBR and SPR, 
and we will also use it when proving the guarantee of \momp{}.   
A special case of the lemma is implicitly used in~\citep{elenberg2018restricted}, 
but we here state and prove it clearly for completeness and convenience. 
\begin{alem}
	\label{a_lem:fts}
	For any $\As\subseteq \dset$, 
	let $\bb{\As}\coloneqq \argmax_{\supp(\xb)\subseteq \As}\f(\xb)$. 
	For any disjoint $\As,\Bs\subseteq\dset$, 
	if $\f$ is 
	$\mab$-RSC and $\tMab$-RSM, 
	we have
	\[\frac{1}{2\tMab}\|\nabla f(\bb{\As})_{\Bs} \|^2_2\le\Fdel{\Bs}{\As}\le\frac{1}{2\mab}\|\nabla f(\bb{\As})_{\Bs} \|^2_2.\]
\end{alem}

\begin{proof}
	We show the first inequality.
	Since $\bb{\As\cup\Bs}$ is the maximizer of $\f$ over
	$\{\xb\in\R^\dset\relmid{}\supp(\xb)\subseteq\As\cup\Bs\}$,
	we have
	$\f(\bb{\As\cup\Bs})\ge \f(\wb+\bb{\As})$
	for any $\supp(\wb)\subseteq \Bs$.
	Therefore, from inequality~\eqref{a_def:rscrsm}, we obtain
	\begin{align}
	\Fdel{\Bs}{\As}
	=
	\f(\bb{\As\cup\Bs}) - \f(\bb{\As})
	\ge{}
	\f(\wb+\bb{\As}) - \f(\bb{\As}) 
	\ge{} 
	\iprod{\nabla \f(\bb{\As})}{\wb} -
	\frac{\tMab }{2}\| \wb \|_2^2.
	\end{align}
	Setting $\wb_\Bs=\frac{1}{\tMab}\nabla \f(\bb{\As})_{\Bs}$ and $\wb_{\dset\bs\Bs}=0$,
	we obtain the first inequality:
	\begin{align}
	\Fdel{\Bs}{\As}
	\ge{}&
	\frac{1}{2\tMab}\|\nabla f(\bb{\As})_{\Bs} \|^2_2.
	\end{align}
	
	We then prove the second inequality.
	Thanks to inequality~\eqref{a_def:rscrsm},
	we have
	\begin{align}
	\Fdel{\Bs}{\As}
	={} 
	\f(\bb{\As\cup\Bs}) - \f(\bb{\As})	
	\le{} 
	\iprod{\nabla \f(\bb{\As})}{\bb{\As\cup\Bs}-\bb{\As}}
	-\frac{\mab}{2}\|\bb{\As\cup\Bs}-\bb{\As}  \|^2_2.
	\end{align}
	Let $\wb\in\R^\dset$ be a vector such that 
	$\supp(\wb)\subseteq \As\cup\Bs$. 
	We consider replacing $\bb{\As\cup\Bs}$ in RHS
	with
	$\wb+\bb{\As}$ 
	and maximizing RHS w.r.t. $\wb$;
	we thus obtain an upper bound of $\Fdel{\Bs}{\As}$ as follows:
	\begin{align}
	\Fdel{\Bs}{\As}
	\le{}&
	\max_{\supp(\wb)\subseteq \As\cup\Bs}\
	\iprod{\nabla \f(\bb{\As})}{\wb}
	-\frac{\mab}{2}\|\wb \|^2_2.
	\end{align}
	The maximum is attained with
	$\wb_{\As\cup\Bs}=\frac{1}{\mab}\nabla \f(\bb{\As})_{\As\cup\Bs}$,
	and so we obtain
	\[
	\Fdel{\Bs}{\As} \le \frac{1}{2\mab}\|\nabla \f(\bb{\As})_{\As\cup\Bs}\|^2_2
	=
	\frac{1}{2\mab}\|\nabla \f(\bb{\As})_{\Bs}\|^2_2,
	\]
	where the last equality comes from the first-order optimality condition
	(or the KKT condition with the linear independence constraint qualification)
	at $\bb{\As}$:
	$\nabla\f(\bb{\As})_\As=0$.
\end{proof}

By using this lemma, we can show that SBR and SPR can be lower bounded 
by ratios of RSC and RSM constants. 
The lower bound of SBR is adopted from~\citep{elenberg2018restricted}, 
and that of SPR improves the existing one presented in~\citep{bogunovic2018robust}. 
\begin{aprop}
	\label{a_prop:brpr}
	For any $\Us\subseteq\dset$ and
	$\s\in\Z_{>0}$, 
	SBR $\brus$ and SPR $\prus$ of 
	$\F(\Ss) = \ls(0) - \min_{\supp(\xb)\subseteq\Ss}\ls(\xb)$ 
	($\forall \Ss\subseteq\dset$)
	are bounded
	with RSC and RSM constants of $\ls$
	as follows:
	\begin{align}
	\brus
	\ge\frac{\m{|\Us|+\s}}{\kM{|\Us|+1}{1}}
	\ge\frac{\m{|\Us|+\s}}{\M{|\Us|+\s}}
	=
	\frac{1}{\kappa_{|\Us|+\s}} 
	& & 
	\text{and}
	& &
	\prus
	\ge\frac{\m{|\Us|+1}}{\kM{|\Us|+\s}{\s}}
	\ge\frac{\m{|\Us|+\s}}{\M{|\Us|+\s}}
	=
	\frac{1}{\kappa_{|\Us|+\s}}.
	\end{align}
\end{aprop}

\begin{proof}
	We refer readers to~\citep{elenberg2018restricted}
	for the proof of the lower bound of $\brus$.
	Here, we show how to obtain the lower bound of $\prus$. 
	From the definition of SPR, 
	we have 
	\begin{align}
	\prus\coloneqq
	\min_{
		{\small \begin{array}{l}
			\Ls,\Ss:\Ls\cap \Ss=\emptyset, \\
			\Ls\subseteq\Us, |\Ss|\le \s
			\end{array}}
	}
	\frac{\Fdel{\Ss}{\Ls}}{\sum_{j\in \Ss}\Fdel{j}{\Ls}},
	\end{align}
	where we regard $0/0=1$. Therefore, we obtain 
	\begin{align}
	\prus
	\ge{}&
	\min_{
		{\small \begin{array}{l}
			\Ls,\Ss:\Ls\cap \Ss=\emptyset, \\
			\Ls\subseteq \Us, |\Ss|\le \s
			\end{array}}
	}
	\frac{\|\nabla f(\bb{\Ls})_{\Ss} \|^2_2}{2\tM{\Ls,\Ls\cup\Ss}}
	\left(
	\sum_{j\in \Ss}
	\frac{|\nabla f(\bb{\Ls})_{j} |^2}{2\m{\Ls,\Ls\cup\{j\}}}
	\right)^{-1}
	& &\cuz{Lemma~\ref{a_lem:fts}}
	\\
	\ge{}&
	\min_{
		{\small \begin{array}{l}
			\Ls,\Ss:\Ls\cap \Ss=\emptyset, \\
			\Ls\subseteq \Us, |\Ss|\le \s
			\end{array}}
	}
	\frac{\m{|\Us|+1}}{\tM{\Ls,\Ls\cup\Ss}}
	\cdot
	\frac{\|\nabla f(\bb{\Ls})_{\Ss} \|^2_2}{\sum_{j\in \Ss}|\nabla f(\bb{\Ls})_{j} |^2}
	& &\cuz{$\m{\Ls,\Ls\cup\{j\}}\ge\m{|\Us|+1}$}
	\\
	\ge{}&
	\min_{
		{\small \begin{array}{l}
			\Ls,\Ss:\Ls\cap \Ss=\emptyset, \\
			\Ls\subseteq \Us, |\Ss|\le \s
			\end{array}}
	}
	\frac{\m{|\Us|+1}}{\tM{\Ls,\Ls\cup\Ss}}
	& &\cuz{$\|\nabla\f(\bb{\Ls})_\Ss\|_2^2=\sum_{j\in \Ss}|\nabla\f(\bb{\Ls})_j|^2$}
	\\
	\ge{}&
	\frac{\m{|\Us|+1}}{\kM{|\Us|+\s}{\s}}.
	& &\cuz{$\tM{\Ls,\Ls\cup\Ss}\le\kM{|\Us|+\s}{\s}$}
	\end{align}
	The proof is completed with $\kM{|\Us|+\s}{\s}\le\M{|\Us|+\s}$ and $\m{|\Us|+1}\ge\m{|\Us|+\s}$.
\end{proof}

\paragraph{Example with Unbounded Curvature} 
We show that 
there is an 
$\ell_0$-constrained minimization 
instance that satisfies the following conditions: 
SBR and SPR of 
$\F(\Ss)
=\ls(0) - \min_{\supp(\xb)\subseteq\Ss}\ls(\xb)$ 
are bounded by a constant, 
while its curvature $\curv$ and inverse curvature $\icurv$ 
are unbounded (i.e., $\curv=\icurv=1$). 
We define 
\begin{align}
\B\coloneqq
\begin{bmatrix}
1 & 1 \\
0 & 1
\end{bmatrix},
& & 
\ab_1 \coloneqq 
\begin{bmatrix}
0 \\ 1
\end{bmatrix},
& & 
\text{and}
& & 
\ab_2 \coloneqq 
\begin{bmatrix}
1 \\ 1
\end{bmatrix}.
\end{align}
Note that we have
\begin{align}
\min_{x_1,x_2\in\R} 
\left\|
\B
\begin{bmatrix}
x_1\\x_2
\end{bmatrix}
-\ab_1
\right\|_2^2
=0,
& & 
\min_{x_1\in\R} 
\left\|
\B
\begin{bmatrix}
x_1\\0
\end{bmatrix}
-\ab_1
\right\|_2^2
=1, 
& & &
\min_{x_2\in\R} 
\left\|
\B
\begin{bmatrix}
0\\x_2
\end{bmatrix}
-\ab_1
\right\|_2^2
=1/2, 
\\
\min_{x_3,x_4\in\R} 
\left\|
\B
\begin{bmatrix}
x_3\\x_4
\end{bmatrix}
-\ab_2
\right\|_2^2
=0,
& & 
\min_{x_3\in\R} 
\left\|
\B
\begin{bmatrix}
x_3\\0
\end{bmatrix}
-\ab_2
\right\|_2^2
=1, 
& & &
\min_{x_4\in\R} 
\left\|
\B
\begin{bmatrix}
0\\x_4
\end{bmatrix}
-\ab_2
\right\|_2^2
=0. 
\end{align}
We define the loss function as 
$\ls(\xb)\coloneqq\|\A\xb-\yb\|_2^2$, 
where $\A\in\R^{\dset\times\dset}$ 
is a block-diagonal matrix 
and $\yb\in\R^\dset$ 
is a vector defined as 
\begin{align}
\A
\coloneqq
\begin{bmatrix}
\B &    &   &  &\\
& \B &   &  &\\
&    & 1 &  &\\
&    &   & \ddots & \\
&    &   &        & 1
\end{bmatrix} 
& & 
\text{and}
& & 
\yb\coloneqq
\begin{bmatrix}
\ab_1 \\
\ab_2 \\
0 \\
\vdots \\
0
\end{bmatrix},
\end{align}
respectively. 
We let $\F(\Ss)= \ls(0) - \min_{\supp(\xb)\subseteq\Ss}\ls(\xb)$ 
for any $\Ss\subseteq\dset$. 
Then we have 
\begin{align}
\Fdel{\{1\}}{\{2\}} = 1/2,
& & 
\F(\{1\}) = 0, 
& &
\Fdel{\{3\}}{\{4\}} = 0, 
& & 
\text{and}
& & 
\F(\{3\}) = 1.
\end{align}
Since $\curv,\icurv\in[0,1]$ 
must satisfy 
\begin{align}
\F(\{1\}) \ge (1-\icurv)\Fdel{\{1\}}{\{2\}}
& & 
\text{and}
& &
\Fdel{\{3\}}{\{4\}} \ge (1-\curv)\F(\{3\}),
\end{align} 
we have $\curv=\icurv=1$. 
On the other hand, the condition number, 
$\kappa$, 
of $\ls$ 
is bounded from above by 
the ratio of 
the largest and smallest eigenvalues of $\A^\top\A$, 
which 
are equal to 
$\frac{3+\sqrt{5}}{2}$ and 
$\frac{3-\sqrt{5}}{2}$, 
respectively; 
hence $\kappa\le\frac{3+\sqrt{5}}{3-\sqrt{5}}$. 
Therefore, 
thanks to Proposition~\ref{a_prop:brpr}, 
we have
$\brus\ge\frac{3-\sqrt{5}}{3+\sqrt{5}}$ 
and 
$\prus\ge\frac{3-\sqrt{5}}{3+\sqrt{5}}$ 
for any $\Us$ and $\s$. 

\subsection{LP with a Cardinality Constraint}
\paragraph{Lower Bounds of SBR and SPR}
As in \Cref{sec:applications}, 
SPR is lower bounded as $\prus\ge1/\s$. 
Furthermore, as shown in~\citep{bian2017guarantees}, 
SBR is lower bounded by some $\br_0>0$ 
under the non-degeneracy assumption: 
For any $\Ss\subseteq\dset$, 
any basic feasible solution of the corresponding LP in the standard form 
is non-degenerate.  

\paragraph{Example with Unbounded Curvature}
An example with $\curv=1$ is provided in \Cref{sec:applications}. 
Note that the example instance satisfies the non-degeneracy assumption. 


\subsection{Coverage Maximization}\label{a_subsec:coverage}
\paragraph{Lower Bounds of SBR and SPR}
Recall that the coverage function is defined as 
$\F(\Ss)\coloneqq\sum_{v\in\ISs} w_v$, 
where 
$w_v\ge0$ ($v\in V$), $\Iset_j\subseteq V$ ($j\in\dset$), 
and 
$\ISs\coloneqq \bigcup_{j\in\Ss} \Iset_j$ for any $\Ss\subseteq\dset$. 
Since the function is submodular, 
we have $\brus=1$ for any $\Us$ and $\s$. 
As assumed in \Cref{sec:applications}, 
any collection of up to $\s$ groups 
covers any $v\in V$ at most $\bo_\s$ times; 
i.e., 
$\bo_\s\coloneqq\max_{v\in V, |\Ss|\le\s}|\{j\in\Ss \relmid{} v\in\Iset_j \}| $.
Therefore, 
\begin{align}
\frac{\Fdel{\Ss}{\Ls}}{\sum_{j\in \Ss}\Fdel{j}{\Ls}}
=
\frac{\sum_{v\in\ISLs\bs\ILs} w_v}{\sum_{j\in \Ss}\sum_{v\in\ILjs\bs\ILs} w_v}
=
\frac{\sum_{v\in\ISLs\bs\ILs} w_v}{\sum_{v\in\ISLs\bs\ILs} w_v|\{j\in\Ss \relmid{} v\in\Iset_j \}|}
\ge
\frac{1}{\bo_\s} 
\end{align} 
holds for any disjoint $\Ls,\Ss\subseteq\dset$ such that $|\Ss|\le\s$,   
which implies $\prus\ge1/\bo_\s$ for any $\Us$ and $\s$.

\paragraph{Example with Unbounded Curvature}
We  provide an example of a coverage function with 
bounded SPR $\prus$ and unbounded curvature $\curv=1$. 
Let 
$V=\{v_1,v_2,v_3\}$ and $w_v=1$ ($v\in V$);  
i.e., $\F(\Ss)=|\ISs|$.  
We let $\d=3$ and define 
$\Iset_1 = \{v_1,v_2\}$,   
$\Iset_2 = \{v_2,v_3\}$, and  
$\Iset_3 = \{v_1,v_3\}$.  
Since each $v\in V$ is covered by at most two groups, 
we have $\bo_\s=2$ for any $\s$, 
which implies $\prus\ge1/2$ for any $\Us$ and $\s$.  
On the other hand, we have 
$\Fdel{\{1\}}{\{2,3\}} = 0$ and $\F(\{1\})=2$, 
which leads to $\curv=1$ since $\curv$ must satisfy 
$\Fdel{j}{\Ss}\ge (1-\curv)\F(j)$ for any $\Ss\subseteq\dset$ and $j\notin\Ss$.

\clearpage

\section{Multi-stage Algorithms}\label{a_sec:multi}

We prove the theoretical guarantees of the multi-stage algorithms 
in \Cref{a_subsec:multi_guarantee}, 
and 
we present experimental results with well- and ill-conditioned synthetic $\ell_0$-constrained instances 
in \Cref{a_subsec:experiments_synthetic}. 

\begin{algorithm}[htb]
	\caption{Multi-stage algorithm}
	\label{a_alg:multi}
	\begin{algorithmic}[1]
		\State $\Us\gets \dset$, $\Ss\gets\emptyset$
		\For{$i=1,\dots,\mm$}
		\State $\Bs_i\gets\argmax_{\Bs\subseteq \Us:|\Bs|\le b_i} 
		\Gs(\Bs)$
		\State $\Ss\gets\Ss\cup\Bs_i$
		\State $\Us\gets \Us\bs\Bs_i$
		\EndFor
		\State \Return $\Ss$
	\end{algorithmic}
\end{algorithm}

\subsection{Theoretical Guarantees}\label{a_subsec:multi_guarantee}
Note that the surrogate functions, $\Gs$, considered below are monotone, 
which means we have $|\Bs_i|=b_i$ in each $i$-th iteration. 
Let $\Ss_i=\Bs_1\cup\cdots\cup\Bs_i$ 
for $i\in[\mm]$ and $\Ss_0=\emptyset$. 
We take $\Sso$ and $\xb^*$, 
which satisfy $\ko=|\Sso|=\|\xb^*\|_0$, 
to be target solutions for WMM 
and $\ell_0$-constrained minimization, respectively.  
As is usual with the proof of greedy-style algorithms, 
we obtain approximation guarantees from 
a lower bound of the marginal gain in each iteration 
as in the following lemma:  
\begin{alem}\label{a_lem:Bi}
	Given any  
	$\app_1,\dots,\app_m$ 
	such that $\app_i\in[0,1]$ ($i\in[m]$), 
	if we can find $\Bs_i\subseteq\dset$ such that 
	$b_i=|\Bs_i|\le\ko$ and 
	\begin{align}\label{a_eq:Bi}
	\Fdel{\Bs_i}{\Ss_{i-1}} 
	\ge 
	\app_i \frac{b_i}{\ko}
	(\F(\Sso) - \F(\Ss_{i-1})) 
	\end{align} 
	in each $i$-th iteration ($i\in[\mm]$), 
	then the following inequality holds:   
	\[
	\F(\Ss_{\mm}) 
	\ge
	\left(
	1 - \prod_{i=1}^{\mm} \left( 1 - \app_i \frac{b_i}{\ko}\right)  
	\right) 
	\F(\Sso) 
	\ge 
	\left(1-\exp\left( -\frac{1}{\ko}\sum_{i=1}^\mm\app_ib_i\right)  \right)
	\F(\Sso). 
	\]
\end{alem}

\begin{proof}
	We first prove that 
	\begin{align}\label{a_eq:prod}
	\F(\Ss_{i}) 
	\ge
	\left(
	1 - \prod_{\ip=1}^{i} \left( 1 - \app_\ip \frac{b_\ip}{\ko}\right)  
	\right) 
	\F(\Sso).
	\end{align} 
	holds for $i=1,\dots,\mm$ by induction. 
	If $i=1$, 
	the inequality holds due to~\eqref{a_eq:Bi}. 
	Assume that we have 
	\begin{align}\label{a_eq:assump}
	\F(\Ss_{i-1}) 
	\ge
	\left(
	1 - \prod_{\ip=1}^{i-1} \left( 1 - \app_\ip \frac{b_\ip}{\ko}\right)  
	\right) 
	\F(\Sso).
	\end{align}	
	Then we obtain 
	\begin{align}
	\F(\Ss_{i})
	&\ge 
	\app_i \frac{b_i}{\ko}
	(\F(\Sso) - \F(\Ss_{i-1}))
	+
	\F(\Ss_{i-1})
	& & \cuz{\eqref{a_eq:Bi}}
	\\
	&\ge
	\app_i \frac{b_i}{\ko}\F(\Sso)
	+
	\left( 1 - \app_i \frac{b_i}{\ko} \right) 
	\left(
	1 - \prod_{\ip=1}^{i-1} \left( 1 - \app_\ip \frac{b_\ip}{\ko}\right)  
	\right) 
	\F(\Sso)
	& & \cuz{\eqref{a_eq:assump}}
	\\
	& = 
	\left(
	1 - \prod_{\ip=1}^{i} \left( 1 - \app_\ip \frac{b_\ip}{\ko}\right)  
	\right) 
	\F(\Sso).
	\end{align} 
	Therefore, 
	the above inequality holds for any $i\in[\mm]$ 
	by induction. 
	By setting $i=\mm$, 
	we obtain the first inequality in Lemma~\ref{a_lem:Bi}. 
	We then prove the second inequality.  
	Since 
	$\app_i \frac{b_i}{\ko}\in[0,1]$ ($i\in[\mm]$), 
	the arithmetic mean of $1-\app_1 \frac{b_1}{\ko},\dots,1-\app_\mm \frac{b_\mm}{\ko}$ 
	is always lower bounded by 
	their geometric mean thanks to AM--GM. 
	Therefore, we have
	\[
	\prod_{i=1}^\mm
	\left(1-\app_i \frac{b_i}{\ko}\right)
	\le
	\left(1-\frac{1}{\mm}\sum_{i=1}^{\mm}\app_i \frac{b_i}{\ko}\right)^\mm
	\le
	\exp\left( -\sum_{i=1}^{\mm}\app_i \frac{b_i}{\ko} \right).  
	\]
	By plugging this into the first inequality, we obtain the desired result. 
\end{proof}

\subsubsection{Multi-Greedy}

Thanks to Lemma~\ref{a_lem:Bi}, 
we can prove the guarantees of \mgreedy{} as follows: 

\begin{thm}\label{a_thm:mg}
	Let $\bmax$ be an integer satisfying $1\le\bmax\le\ko$. 
	Set $b_1,\dots,b_m$ so as to satisfy  
	$b_i\in[\bmax]$ for $i\in[\mm]$ and $\sum_{i\in[\mm]}b_i=\k$.  
	If 
	$\Ss_\mm$ is the solution 
	obtained with \mgreedy{} 
	and 
	$\F$ is $(\br_{\Ss_{\mm},\ko}, \pr_{\Ss_{\mm},\bmax})$-WM, 
	we have
	\begin{align}
	\F(\Ss_\mm)
	\ge 
	\left(1- 
	\prod_{i=1}^\mm
	\left(1 -
	\br_{\Ss_{i-1},\ko}
	\pr_{\Ss_{i-1},b_i}
	\frac{b_i}{\ko}
	\right)  \right)
	\F(\Sso)
	\ge 
	\left(1-\exp\left( -
	\br_{\Ss_{\mm},\ko}
	\pr_{\Ss_{\mm},\bmax}
	\frac{\k}{\ko}
	\right)  \right)
	\F(\Sso). 
	\end{align}
\end{thm}

\begin{proof}
	To prove the theorem, 
	it suffices that 
	$\Bs_i$ chosen by \mgreedy{} 
	in each iteration 
	satisfies~\eqref{a_eq:Bi} 
	with $\app_i=\br_{\Ss_{i-1},\ko}
	\pr_{\Ss_{i-1},b_i}$; 
	then we can obtain the theorem by using Lemma~\ref{a_lem:Bi} 
	and 
	$\br_{\Ss_{i-1},\ko}
	\pr_{\Ss_{i-1},b_i}
	\ge 
	\br_{\Ss_{\mm},\ko}
	\pr_{\Ss_{\mm},\bmax}$ ($i\in[\mm]$). 
	Note that \mgreedy{} uses 
	$\Gsim(\Bs)=\sum_{j\in\Bs}\Fdel{j}{\Ss_{i-1}}$ 
	as a surrogate function in each $i$-th iteration. 
	From $b_i \le \ko=|\Sso|$ and the greedy rule, 
	we have
	\begin{align}\label{a_eq:avegr}
	\frac{1}{b_i}\sum_{j\in\Bs_i}\Fdel{j}{\Ss_{i-1}}
	\ge 
	\frac{1}{\ko}\sum_{j\in\Ss^*\bs\Ss_{i-1}}\Fdel{j}{\Ss_{i-1}}.
	\end{align}
	Therefore, 
	we obtain
	\begin{align}
	&\Fdel{\Bs_i}{\Ss_{i-1}}
	\\
	\ge{}&
	\pr_{\Ss_{i-1},b_i}\sum_{j\in\Bs_i}\Fdel{j}{\Ss_{i-1}}
	& & \cuz{definition of $\pr_{\Ss_{i-1},b_i}$}
	\\
	\ge{}&
	\pr_{\Ss_{i-1},b_i}
	\frac{b_i}{\ko}
	\sum_{j\in\Sso\bs\Ss_{i-1}}\Fdel{j}{\Ss_{i-1}}
	& & \cuz{\eqref{a_eq:avegr}}
	\\
	\ge{}&
	\br_{\Ss_{i-1},|\Sso\bs\Ss_{i-1}|}
	\pr_{\Ss_{i-1},b_i}
	\frac{b_i}{\ko}
	\Fdel{\Sso\bs\Ss_{i-1}}{\Ss_{i-1}}
	& & \cuz{definition of $\br_{\Ss_{i-1},|\Sso\bs\Ss_{i-1}|}$}
	\\
	\ge{}&
	\br_{\Ss_{i-1},\ko}
	\pr_{\Ss_{i-1},b_i}
	\frac{b_i}{\ko}
	\Fdel{\Sso}{\Ss_{i-1}}
	& & \cuz{$\br_{\Ss_{i-1},|\Sso\bs\Ss_{i-1}|}\ge\br_{\Ss_{i-1},\ko}$ and $\Fdel{\Sso\bs\Ss_{i-1}}{\Ss_{i-1}}=\Fdel{\Sso}{\Ss_{i-1}}$}  
	\\
	\ge{}&
	\br_{\Ss_{i-1},\ko}
	\pr_{\Ss_{i-1},b_i}
	\frac{b_i}{\ko}
	(\F(\Sso) - \F(\Ss_{i-1})).
	& & \cuz{monotonicity}
	\end{align}
	Thus the proof is completed. 
\end{proof}

\subsubsection{Multi-OMP}
We then prove the following guarantee of \momp{}. 

\begin{thm}\label{a_thm:mo}
	Suppose that 
	$\F$ is defined as 
	$\F(\Ss)=\ls(0) - \min_{\supp(\xb^\prime)\subseteq\Ss}\ls(\xb^\prime)$ 	($\forall\Ss\subseteq\dset$)
	and 
	$b_1,\dots,b_\mm$ are set 
	as in Theorem~\ref{a_thm:mg}. 
	Assume that 
	$\ls$ is $\m{\k+\ko}$-RSC and $\M{\k,\bmax}$-RSM. 
	If $\Ss_m$ is a solution 
	obtained with \momp{}, 
	then we have  
	\begin{align}
	\F(\Ss_\mm)
	&\ge 
	\left(1-
	\prod_{i=1}^\mm
	\left(1 -
	\frac{\m{\Ss_{i-1},\Ss_{i-1}\cup\Sso}}{\tM{\Ss_{i-1},\Ss_i}}
	\frac{b_i}{\ko}
	\right)  \right)
	\F(\Sso)
	\\
	&\ge 
	\left(1-\exp\left( -\frac{1}{\ko}\sum_{i=1}^\mm
	\frac{\m{\Ss_{i-1},\Ss_{i-1}\cup\Sso}}{\tM{\Ss_{i-1},\Ss_i}}
	b_i\right)  \right)
	\F(\Sso)
	\\
	&\ge 
	\left(1-\exp\left( -
	\frac{\m{\Ss_{\mm},\Ss_{\mm}\cup\Sso}}{\max_{i\in[\mm]}\tM{\Ss_{i-1},\Ss_i}}
	\frac{\k}{\ko}
	\right)  \right)
	\F(\Sso) 
	\\
	&\ge 
	\left(1-\exp\left( -
	\frac{\m{\k+\ko}}{\M{\k,\bmax}}
	\frac{\k}{\ko}
	\right)  \right)
	\F(\Sso). 
	\end{align}
	Consequently, 
	solution $\xb=\argmin_{\supp(\xb^\prime)\subseteq\Ss_\mm}\ls(\xb^\prime)$ satisfies 
	\begin{align}
	\ls(\xb)
	&\le{}
	\ls(\xb^*)
	+
	\prod_{i=1}^\mm
	\left(1 -
	\frac{\m{\Ss_{i-1},\Ss_{i-1}\cup\Sso}}{\tM{\Ss_{i-1},\Ss_i}}
	\frac{b_i}{\ko}
	\right)
	(\ls(0) - \ls(\xb^*))
	\\
	&\le{}
	\ls(\xb^*)
	+
	\exp\left( -
	\frac{\m{\k+\ko}}{\M{\k,\bmax}}
	\frac{\k}{\ko}
	\right)
	(\ls(0) - \ls(\xb^*))
	\\
	&\le{}
	\ls(\xb^*)
	+
	\exp\left( -
	\frac{1}{\kappa_{\k+\ko}}
	\frac{\k}{\ko}
	\right)
	(\ls(0) - \ls(\xb^*)).
	\end{align}
\end{thm}
\begin{proof}
	As in Section~\ref{a_subsec:ell_0}, we define 
	$\f(\xb)\coloneqq\ls(0)-\ls(\xb)$ 
	and 
	$\bb{\Ss}\coloneqq\argmax_{\supp(\xb)\subseteq\Ss}\f(\xb)$ 
	for any given $\Ss\subseteq\dset$. 
	Analogous with the proof of Theorem~\ref{a_thm:mg}, 
	we prove \eqref{a_eq:Bi}, 
	where $\app_i=\frac{\m{\Ss_{i-1},\Ss_{i-1}\cup\Sso}}{\tM{\Ss_{i-1},\Ss_i}}$. 
	Note that 
	\momp{} uses 
	$
	\Gsim(\Bs)
	=
	\sum_{j\in\Bs}|\nabla\ls(\bb{\Ss_{i-1}})_j|^2
	=
	\sum_{j\in\Bs}|\nabla\f(\bb{\Ss_{i-1}})_j|^2
	=
	\|\nabla\f(\bb{\Ss_{i-1}})_\Bs\|_2^2
	$ 
	as a surrogate function. 
	Therefore, 
	by using 
	$b_i \le \ko=|\Sso|$, 
	$\nabla\f(\bb{\Ss_{i-1}})_{\Ss_{i-1}}=0$, 
	and 
	the greedy rule, 
	we obtain  
	\begin{align}\label{a_eq:aveomp}
	\frac{1}{b_i}
	\|\nabla\f(\bb{\Ss_{i-1}})_{\Bs_i}\|_2^2
	\ge
	\frac{1}{\ko}
	\|\nabla\f(\bb{\Ss_{i-1}})_{\Sso\bs\Ss_{i-1}}\|_2^2.
	\end{align}
	By using this inequality and Lemma~\ref{a_lem:fts}, 
	we obtain 
	\begin{align}
	\Fdel{\Bs_i}{\Ss_{i-1}}
	&\ge
	\frac{1}{2\tM{\Ss_{i-1},\Ss_i}}
	\|\nabla\f(\bb{\Ss_{i-1}})_{\Bs_i}\|_2^2
	& & \cuz{Lemma~\ref{a_lem:fts}}
	\\
	&\ge
	\frac{1}{2\tM{\Ss_{i-1},\Ss_i}}
	\cdot
	\frac{b_i}{\ko}
	\|\nabla\f(\bb{\Ss_{i-1}})_{\Sso\bs\Ss_{i-1}}\|_2^2
	& & \cuz{\eqref{a_eq:aveomp}}
	\\
	&\ge
	\frac{\m{\Ss_{i-1},\Ss_{i-1}\cup\Sso}}{\tM{\Ss_{i-1},\Ss_i}}
	\cdot
	\frac{b_i}{\ko}
	\Fdel{\Sso\bs\Ss_{i-1}}{\Ss_{i-1}}
	& & \cuz{Lemma~\ref{a_lem:fts}}
	\\
	&\ge
	\frac{\m{\Ss_{i-1},\Ss_{i-1}\cup\Sso}}{\tM{\Ss_{i-1},\Ss_i}}
	\cdot
	\frac{b_i}{\ko}
	(\F(\Sso)-\F(\Ss_{i-1}).
	& & \cuz{monotonicity}
	\end{align}
	Thus the theorem holds thanks to Lemma~\ref{a_lem:Bi}.   
\end{proof}

\subsection{Experiments with Synthetic $\ell_0$-constrained Minimization Instances}\label{a_subsec:experiments_synthetic}	

We here evaluate the multi-stage algorithms 
with synthetic $\ell_0$-constrained minimization instances.


\begin{figure}
	\begin{tabular}{p{.3\textwidth}p{.3\textwidth}p{.3\textwidth}}
		\centering
		\includegraphics[width=1.0\linewidth]{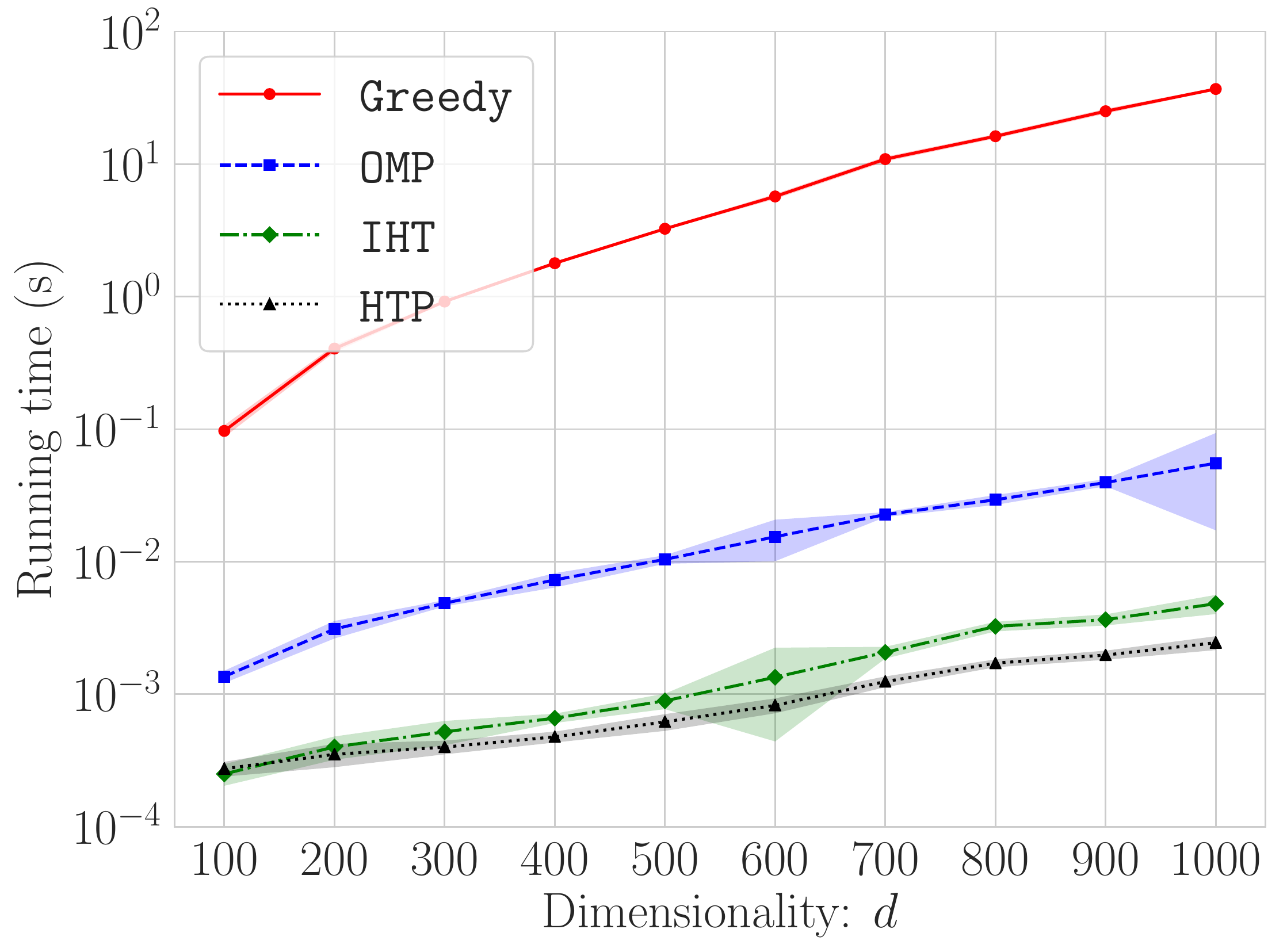}
		\subcaption{$m=\k$, Well-conditioned}
		\label{fig:well_time_1}
		&
		\centering
		\includegraphics[width=1.0\linewidth]{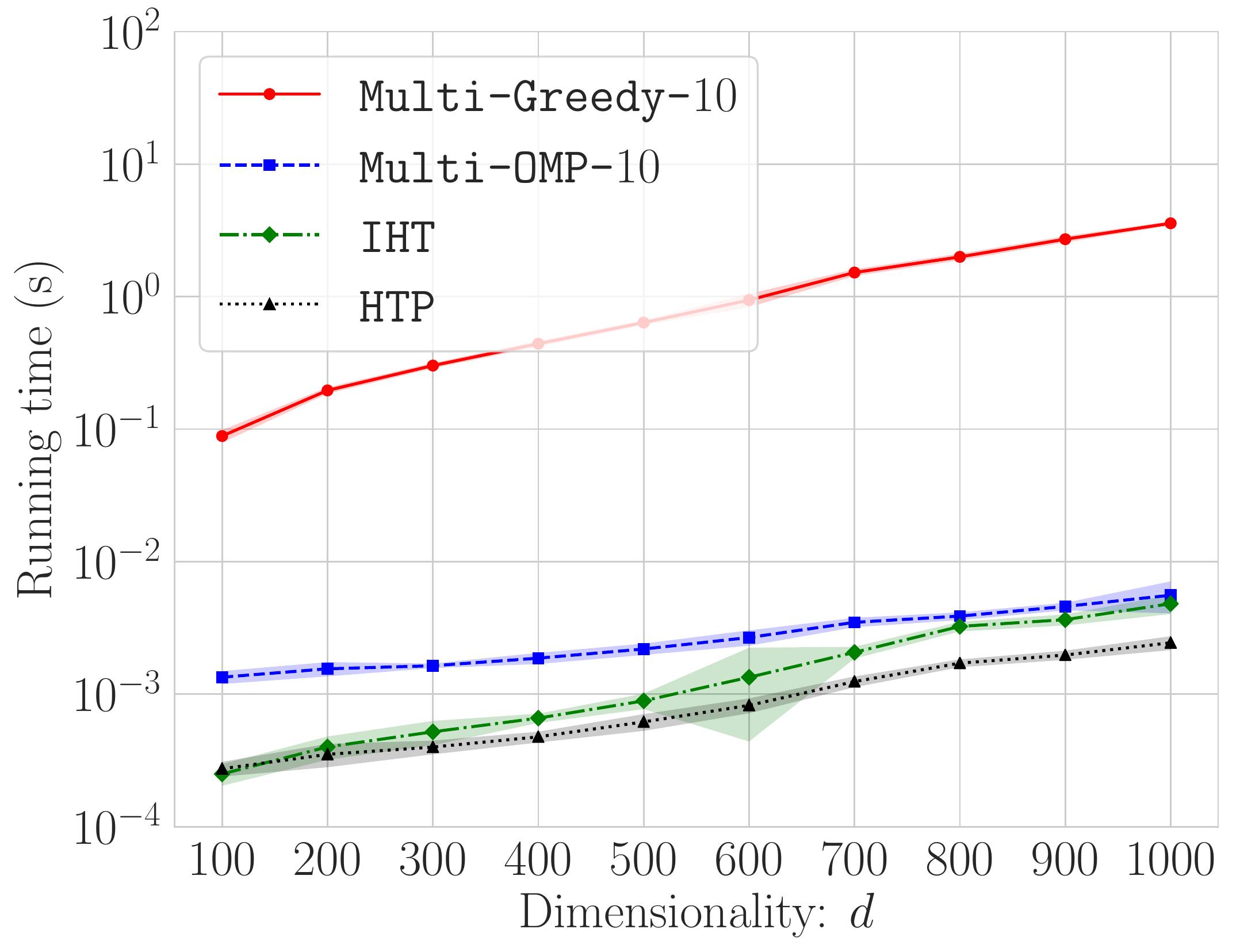}
		\subcaption{$m=10$, Well-conditioned}
		\label{fig:well_time_2}
		&
		\centering
		\includegraphics[width=1.0\linewidth]{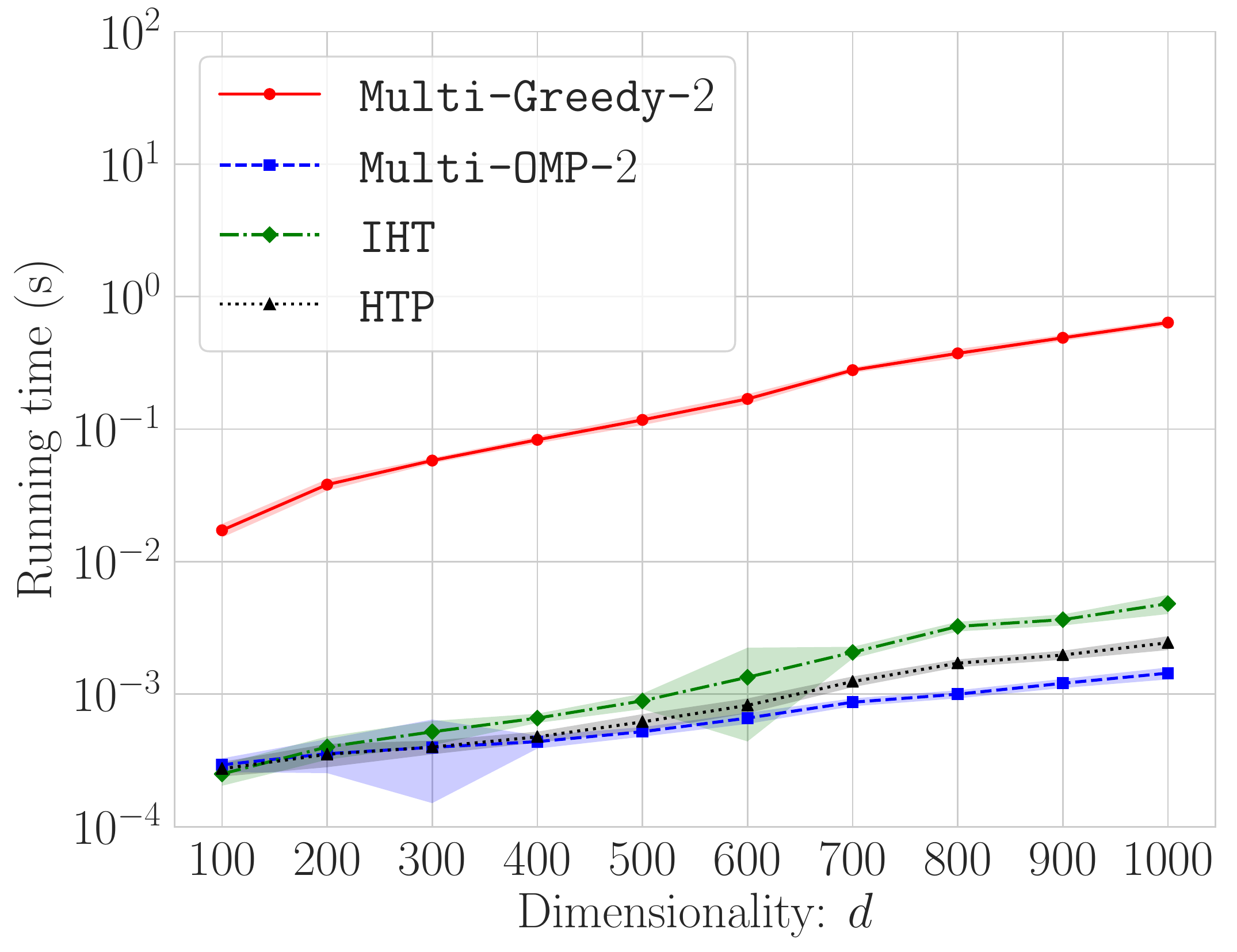}
		\subcaption{$m=2$, Well-conditioned}
		\label{fig:well_time_3}
	\end{tabular}
	\begin{tabular}{p{.3\textwidth}p{.3\textwidth}p{.3\textwidth}}
		\centering
		\includegraphics[width=1.0\linewidth]{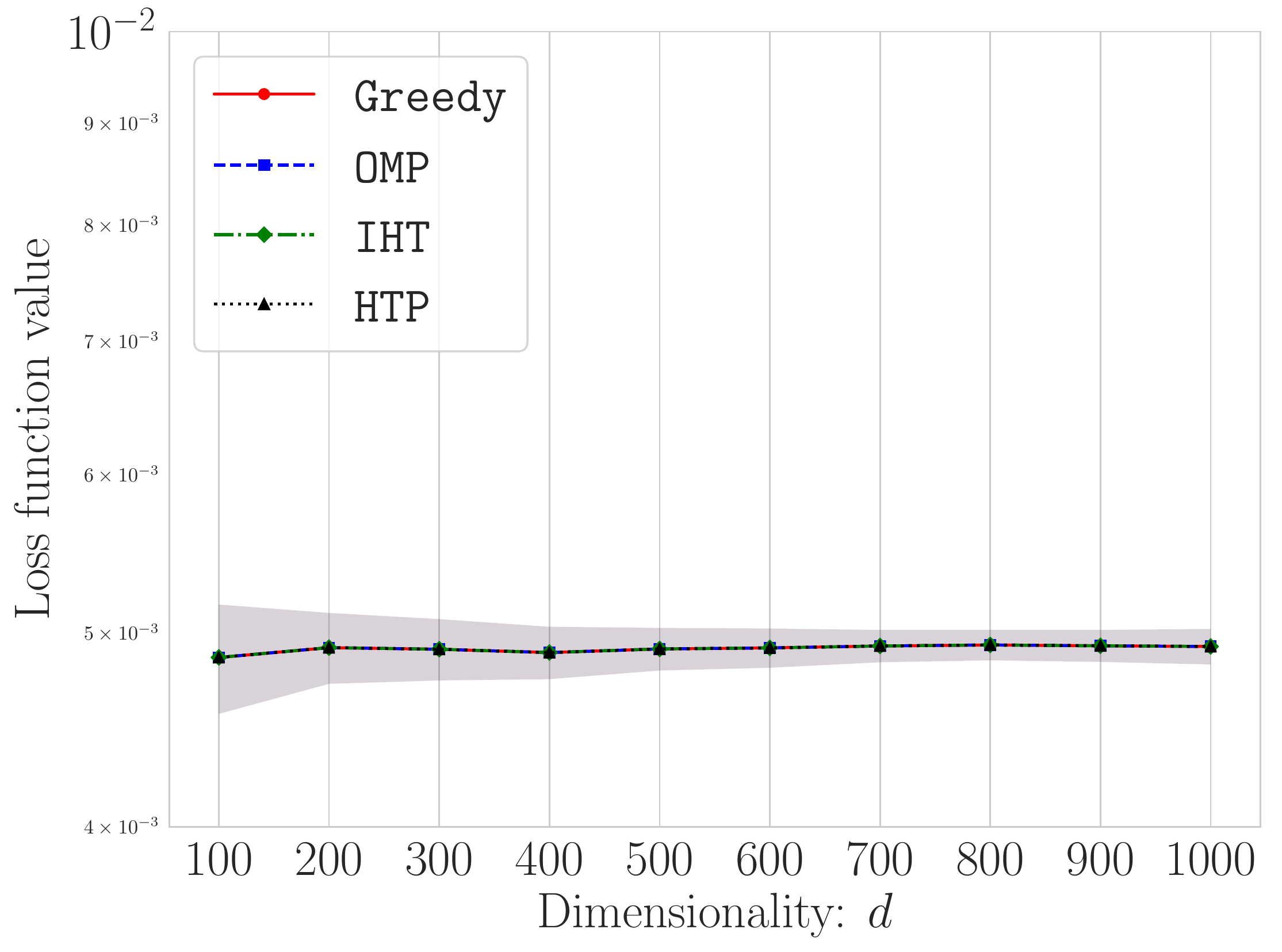}
		\subcaption{$m=\k$, Well-conditioned}
		\label{fig:well_loss_1}
		&
		\centering
		\includegraphics[width=1.0\linewidth]{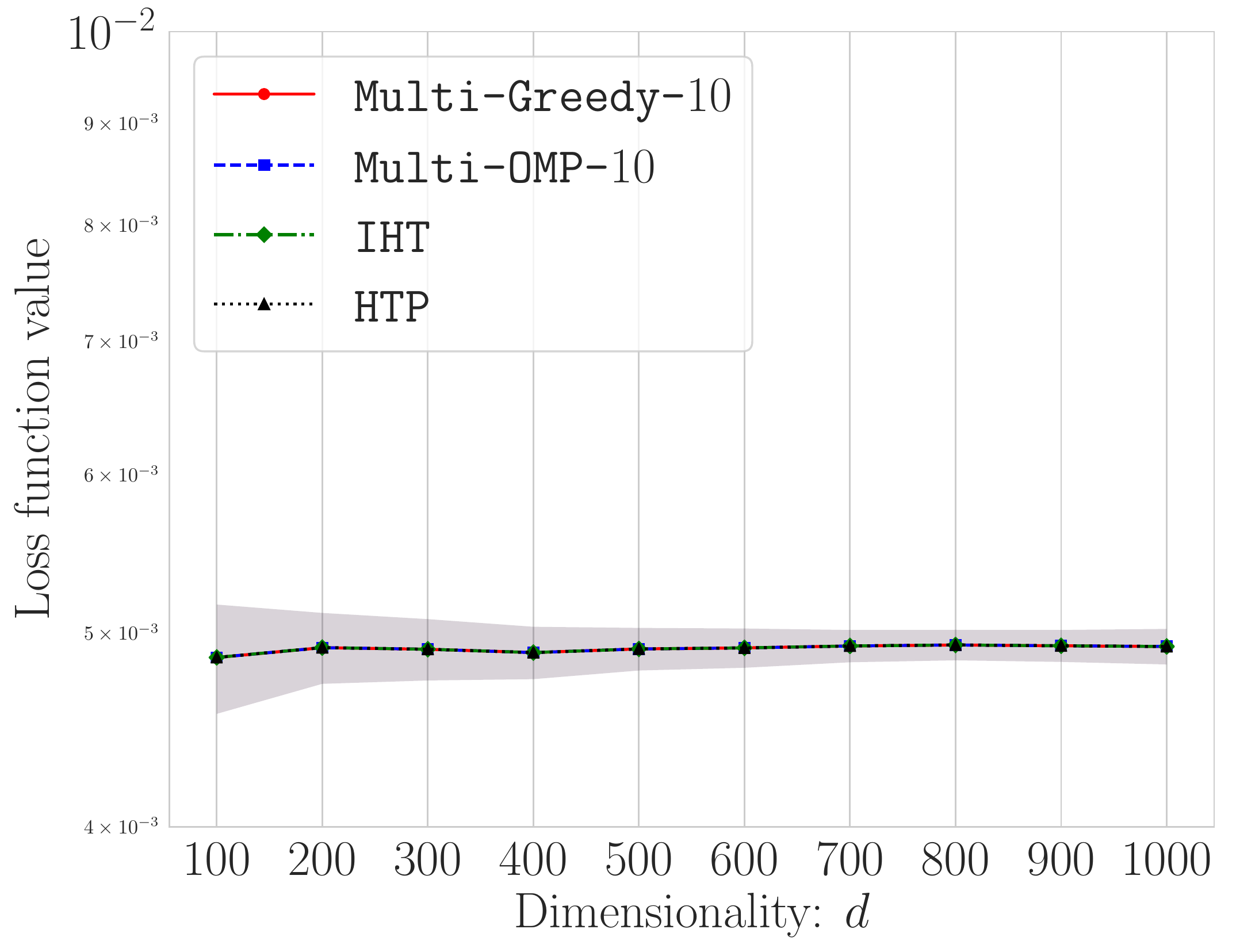}
		\subcaption{$m=10$, Well-conditioned}
		\label{fig:well_loss_2}
		&
		\centering
		\includegraphics[width=1.0\linewidth]{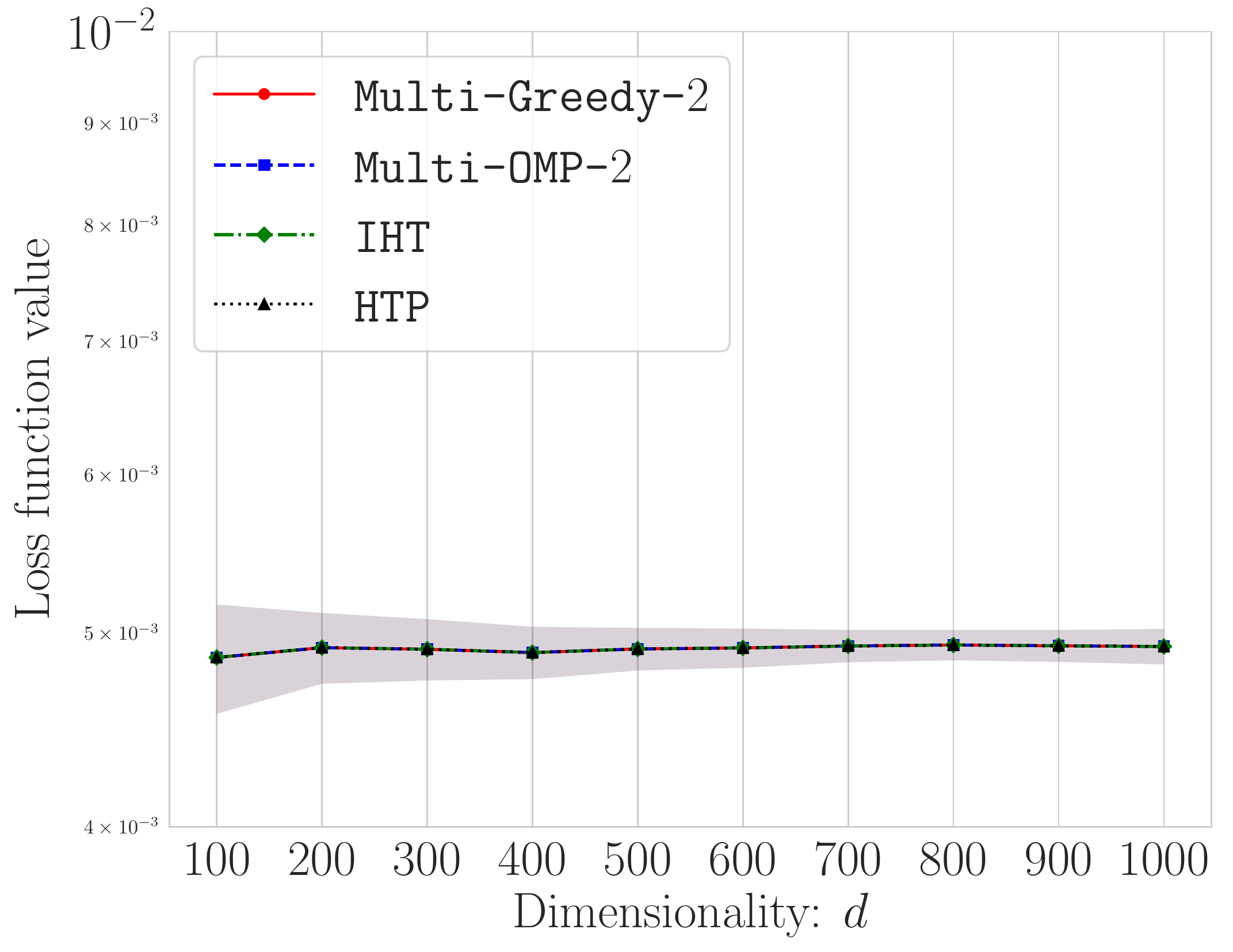}
		\subcaption{$m=2$, Well-conditioned}
		\label{fig:well_loss_3}
	\end{tabular}
	\begin{tabular}{p{.3\textwidth}p{.3\textwidth}p{.3\textwidth}}
		\centering
		\includegraphics[width=1.0\linewidth]{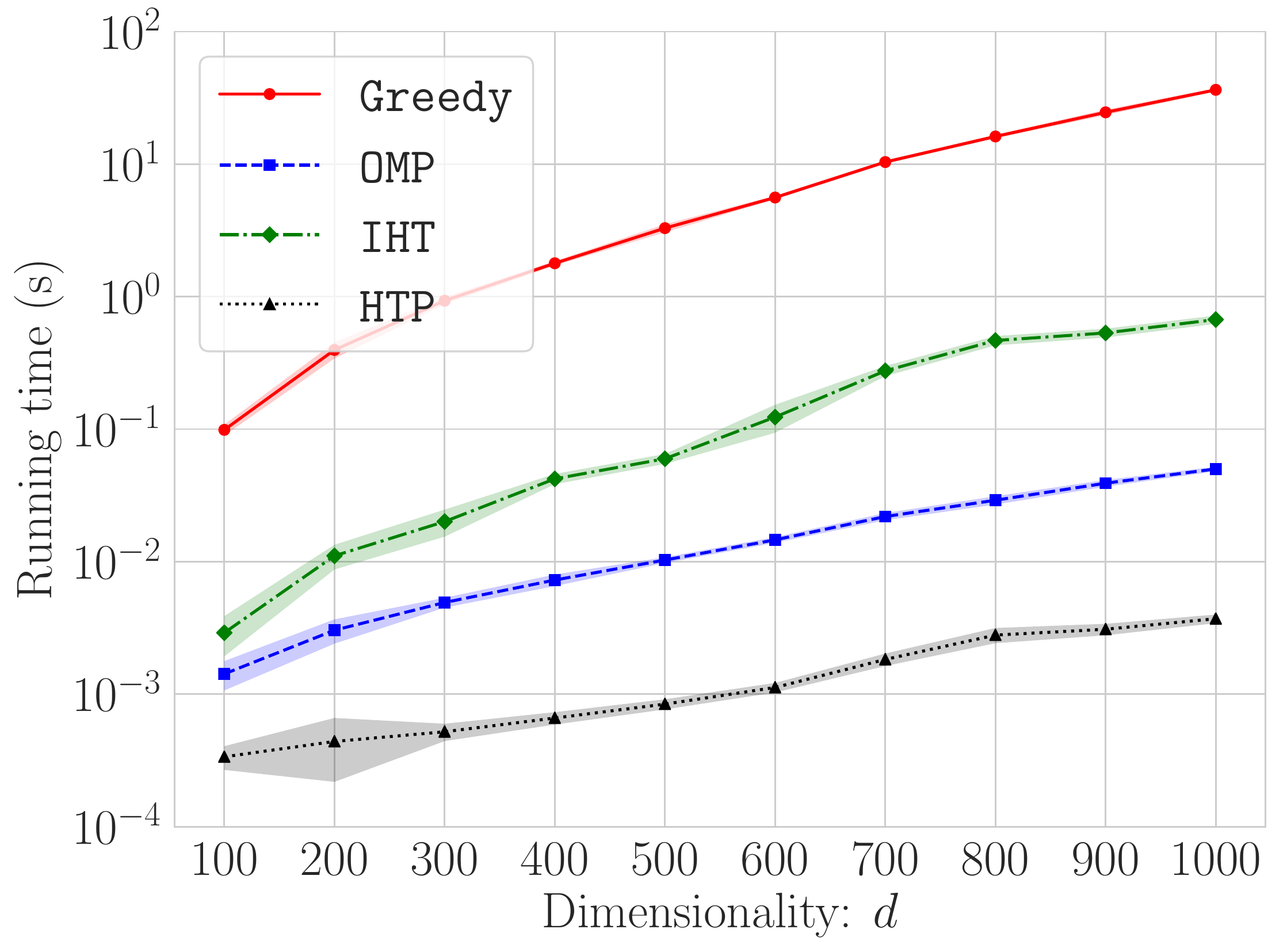}
		\subcaption{$m=\k$, Ill-conditioned}
		\label{fig:ill_time_1}
		&
		\centering
		\includegraphics[width=1.0\linewidth]{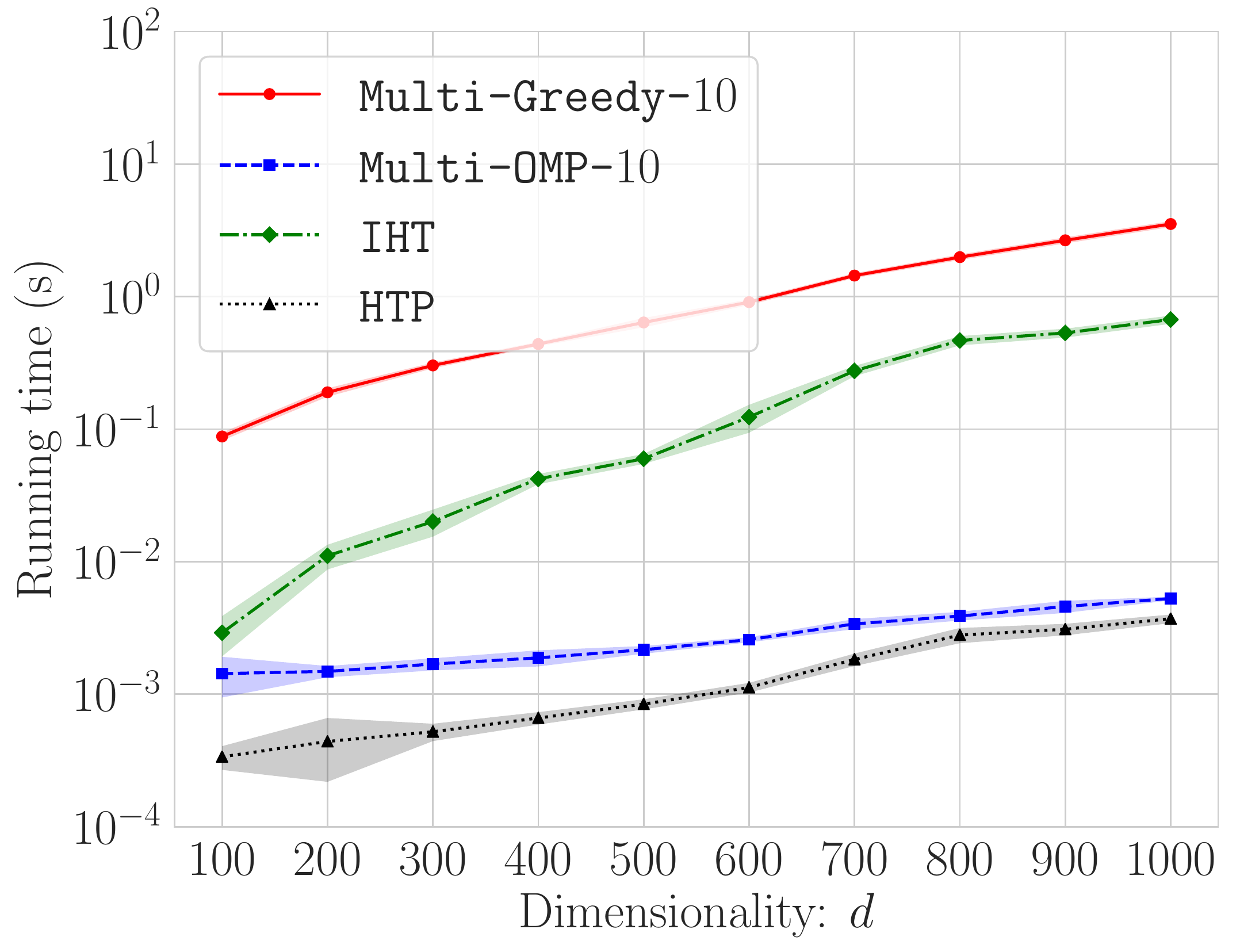}
		\subcaption{$m=10$, Ill-conditioned}
		\label{fig:ill_time_2}
		&
		\centering
		\includegraphics[width=1.0\linewidth]{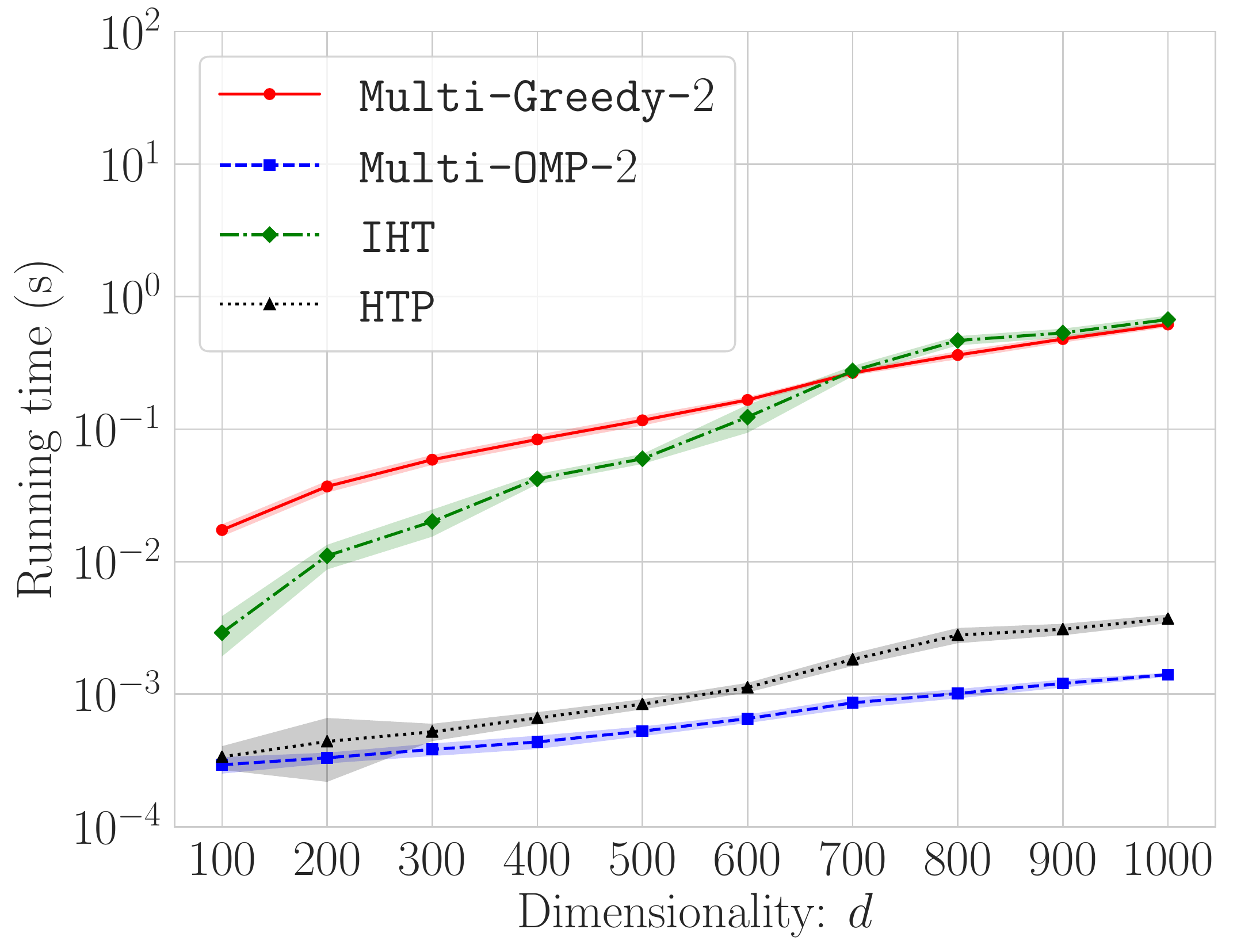}
		\subcaption{$m=2$, Ill-conditioned}
		\label{fig:ill_time_3}
	\end{tabular}
	\begin{tabular}{p{.3\textwidth}p{.3\textwidth}p{.3\textwidth}}
		\centering
		\includegraphics[width=1.0\linewidth]{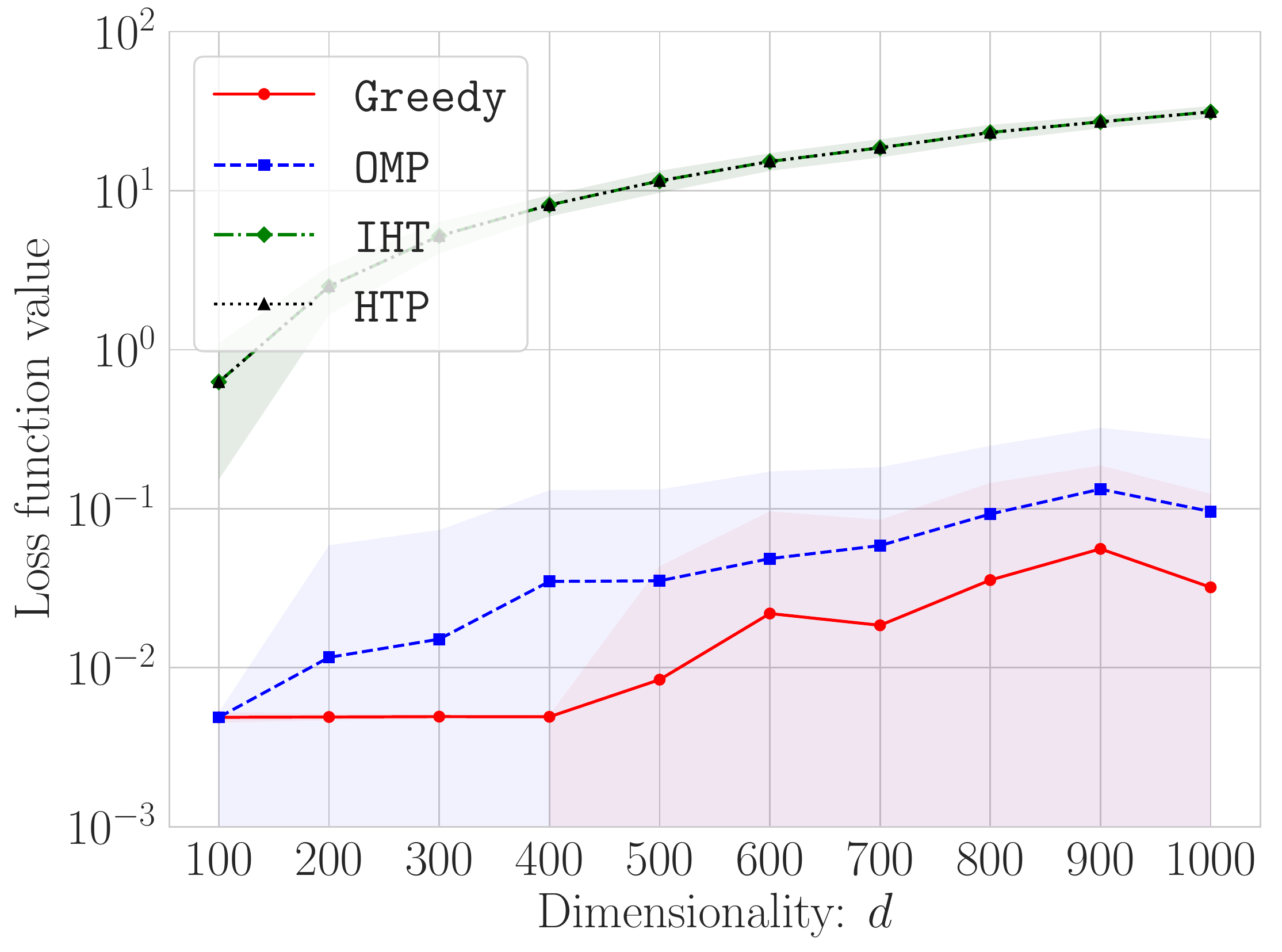}
		\subcaption{$m=\k$, Ill-conditioned}
		\label{fig:ill_loss_1}
		&
		\centering
		\includegraphics[width=1.0\linewidth]{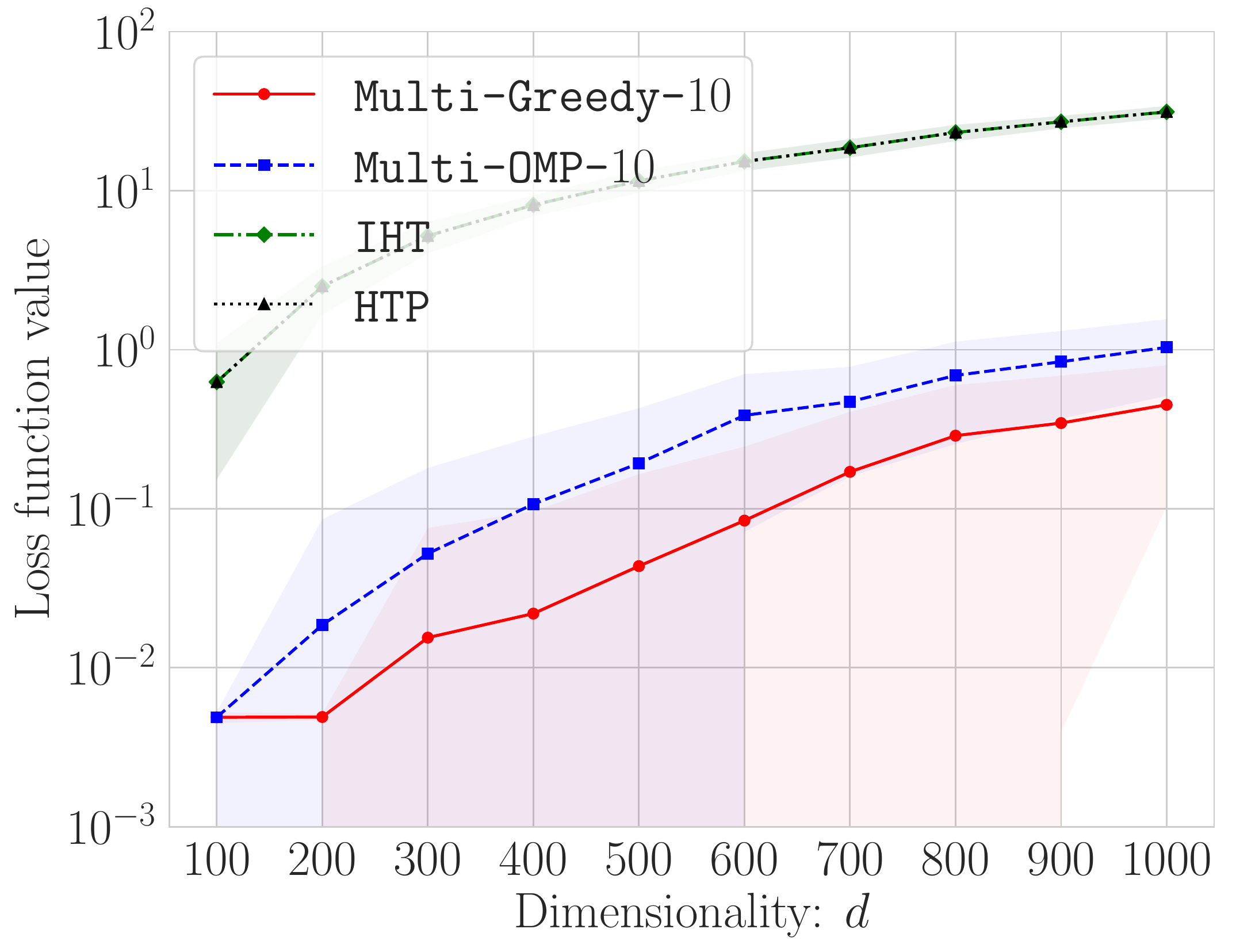}
		\subcaption{$m=10$, Ill-conditioned}
		\label{fig:ill_loss_2}
		&
		\centering
		\includegraphics[width=1.0\linewidth]{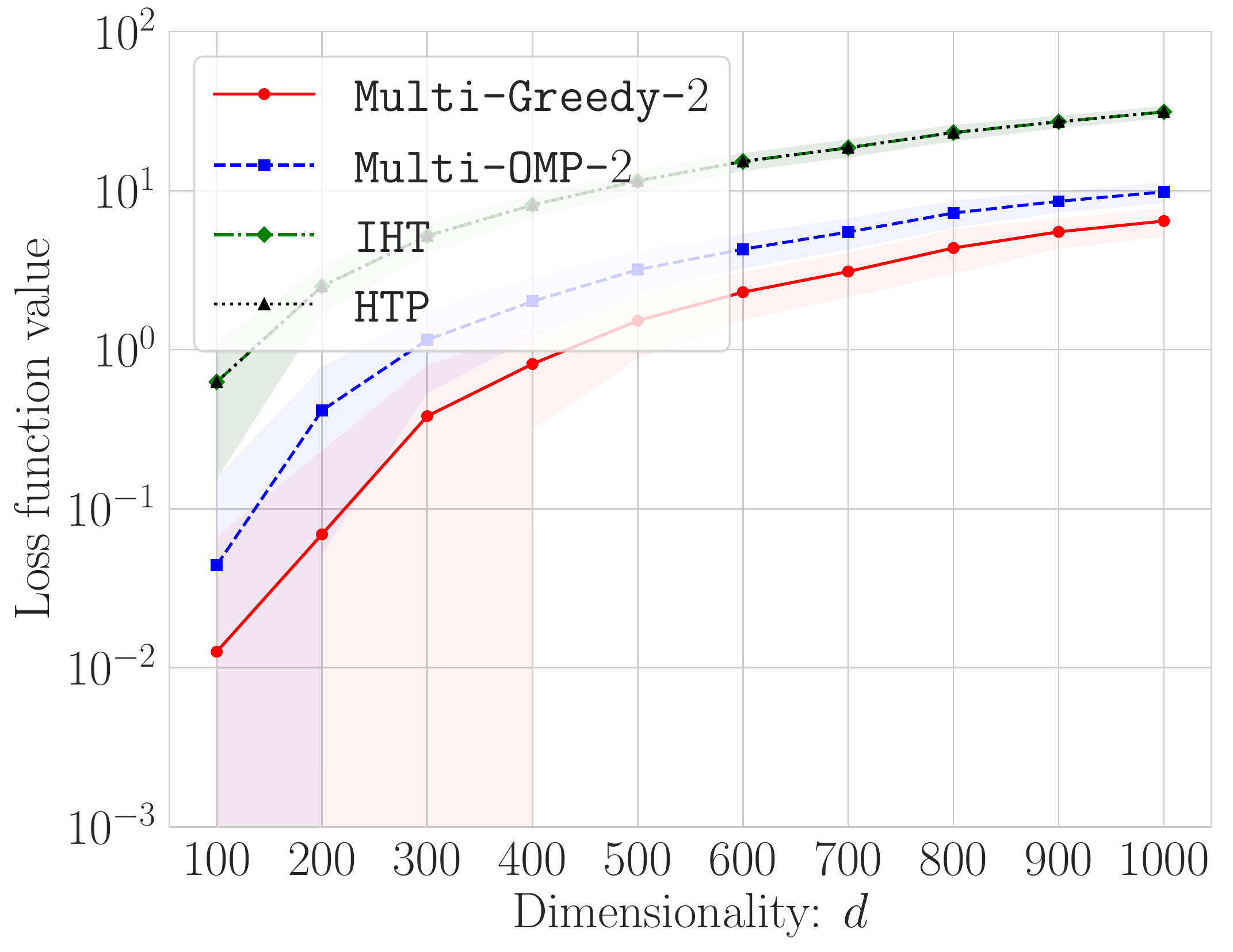}
		\subcaption{$m=2$, Ill-conditioned}
		\label{fig:ill_loss_3}
	\end{tabular}
	\caption{
		Semi-log Plots of Running Times and 
		Loss Function Values for Well- and Ill-conditioned Instances.  
		Figures (a)--(f) and (g)--(l) 
		correspond to well- and 
		ill-conditioned instances, respectively. 
		The left, middle, and right figures 
		show the results with 
		$m=\k, 10,$ and $2$, 
		respectively. 
		Each curve and error band indicate the 
		mean and standard deviation, 
		respectively, calculated over $100$ random instances. 
	}
	\label{fig:syn}
\end{figure}

\paragraph{Settings}
We consider well- and ill-conditioned sparse regression instances. 
Given
design matrix $\A\in\R^{n \times \d}$
and vector
$\yb\in\R^{n}$,
we use the square loss function:
$\ls(\xb)\coloneqq\frac{1}{2n}\|\yb-\A\xb\|_2^2$.
%
We randomly generate well- and ill-conditioned instances as follows:  
We set the first $\k$ entries
of the true sparse solution, $\xtrue$, at 1 and the others at 0. 
In the well-conditioned case, 
we draw each entry of
$\A$ from the standard normal distribution, denoted by $\unif$. 
In the ill-conditioned case, 
we draw each row of $\A$ from a correlated $d$-dimensional normal distribution, whose mean and correlation coefficient are set at $0$ and $0.3$, respectively. 
Then, for both well- and ill-conditioned instances, 
we set $\yb=\A\xtrue + 0.1\ub$, 
where each entry of $\ub\in\R^{n}$ is drawn from $\unif$.
We consider various dimensionalities: $\d=100,200,\dots,1000$. 
We let 
$\k=0.1\d$ and $n=\lfloor10\k\log\d\rfloor$. 
For each $d$ value, we generate $100$ random instances as above. 
We apply the multi-stage algorithms with $m=\k, 10,$ and $2$ iterations 
to the instances, 
where $m=\k$ corresponds to the standard \greedy{}/\omp{}. 
We use the projected-gradient-based methods (\iht{} and \iht{}) as baselines. 
We evaluate these methods in terms of running times and 
loss function values. 

\paragraph{Results}
\Cref{fig:syn} summarizes the results. 
We see that the multi-stage algorithms speed up as $m$ decreases; 
they can become as fast as \iht{}/\htp. 
In the well-conditioned case, 
\momp{-$2$} is the fastest, 
and 
all the methods achieve the same loss function 
values, 
implying that the well-conditioned instances 
are so easy as to be solved almost optimally by all the methods. 
In the ill-conditioned case, 
as mentioned in 
\Cref{subsubsec:real}, 
greedy-style methods achieve 
better loss function values than 
the projected-gradient-based methods. 
We see that the parameter, $m$, 
of multi-stage algorithms  
controls the trade-off 
between the running times and loss function values. 
%
To conclude, 
multi-stage algorithms with appropriate $m$ values 
can outperform projected-gradient-based 
methods both in running time and solution quality, 
particularly when the instances are ill-conditioned.

\clearpage

\section{Fixed-parameter Tractability}\label{a_sec:fpt}
In this section, we prove the guarantee of the randomized FPT algorithm 
for WMM. 
This result is an extension of~\citep{skowron2017fpt}, 
which developed the randomized FPT algorithm for a subclass of 
monotone submodular maximization. 
\begin{algorithm}[htb]
	\caption{Randomized FPT algorithm}
	\label{a_alg:fpt}
	\begin{algorithmic}[1]
		\State Execute \singlerun~$T$ times
		and return the best solution.
		\Function{\singlerun}{}
		\State $\Ss_{0}\gets\emptyset$
		\For{$i=1,\dots,k$}
		\State Choose $j\in\dset\bs\Ss_{i-1}$ randomly
		with probability proportional to $\Fdel{j}{\Ss_{i-1}}$.
		\State $\Ss_{i}\gets\Ss_{i-1}\cup\{j\}$.
		\EndFor
		\State \return~$\Ss_k$
		\EndFunction
	\end{algorithmic}
\end{algorithm}

Let $\Sso$ be an optimal solution; i.e., $\Sso\in\argmax_{\Ss:|\Ss|\le\k}\F(\Ss)$.  
We first prove a key lemma, which provides a lower bound of the probability that
$j\in\Sso$ is chosen in each iteration of \singlerun. 
\begin{alem}\label{a_lem:prob_bound}
	For $i\in[k]$,
	let $\Ss_{i-1}$ be the partial solution
	that is constructed in the loops of \singlerun.
	Then
	the probability $p\in[0,1]$ that newly chosen $j\in\dset\bs\Ss_{i-1}$ is included
	in $\Ss^*$
	is bounded from below as follows:
	\[
	p\ge
	\gbk\cdot
	\frac{\Fdel{\Ss^*}{\Ss_{i-1}}}{\Fdel{\dset}{\Ss_{i-1}}}.
	\]
\end{alem}

\begin{proof}
	The proof is obtained directly from the definitions of 
	SBR and SPR as follows:
	\begin{align}
	p=\frac{\sum_{j\in\Ss^*\bs\Ss_{i-1}}\Fdel{j}{\Ss_{i-1}}}{\sum_{j\in\dset\bs\Ss_{i-1}}\Fdel{j}{\Ss_{i-1}}}
	\ge
	{\br_{\Ss_{i-1},|\Sso\bs\Ss_{i-1}|}}{\pr_{\Ss_{i-1},|\dset\bs\Ss_{i-1}|}}
	\cdot
	\frac{\Fdel{\Ss^*}{\Ss_{i-1}}}{\Fdel{\dset}{\Ss_{i-1}}}
	\ge
	\gbk\cdot
	\frac{\Fdel{\Ss^*}{\Ss_{i-1}}}{\Fdel{\dset}{\Ss_{i-1}}}.
	\end{align}
\end{proof}
Using this lemma we obtain the theorem as follows: 

\begin{thm}\label{a_thm:fpt}
	Assume $\F$ to be ($\brkk, \pr_{\k,\d}$)-WM. 
	Let $\Sso$ be an optimal solution for problem~\eqref{prob:main_F} 
	and 
	$\tilde{\F}\coloneqq\F(\dset)-\F(\Sso)$.   
	For any $\epsilon>0$,
	if 
	$
	T\ge\Tfpt
	$, 
	then \Cref{alg:fpt} returns solution $\Ss$ satisfying
	$\F(\Ss)\ge\F(\Sso)-\epsilon$ with a probability of at least $1-\delta$.
\end{thm}

\begin{proof}
	We consider a single invocation of \singlerun.
	In each $i$-th iteration ($i\in[k]$),
	one of the following two conditions occurs:
	\begin{align}
	\F(\Ss_{i-1}) \ge \F(\Sso) - \epsilon,
	\label{a_eq:one} \\
	\F(\Ss_{i-1}) < \F(\Sso) - \epsilon.
	\label{a_eq:two}
	\end{align}
	Once~\eqref{a_eq:one} occurs
	for some $i\in[k]$,
	then we have
	$\F(\Ss_\k)\ge\F(\Ss_{i-1})\ge\F(\Sso)-\epsilon$
	thanks to the monotonicity of $\F$.
	If~\eqref{a_eq:two} occurs,
	we have 
	\[
	\frac{\Fdel{\Ss^*}{\Ss_{i-1}}}{\Fdel{\dset}{\Ss_{i-1}}}
	\ge
	\frac{\F(\Sso)-\F(\Ss_{i-1})}{\F(\dset)-\F(\Ss_{i-1})}
	=
	\frac{\F(\Sso)-\F(\Ss_{i-1})}{\tilde{\F}
		+ \F(\Sso) - \F(\Ss_{i-1})}
	> \frac{\epsilon}{\tilde{\F}+\epsilon}.
	\]
	Hence, newly chosen $j\in\dset\bs\Ss_{i-1}$
	is included in $\Ss^*$ with probability $p\ge\gbk\cdot
	\frac{\epsilon}{\tilde{\F}+\epsilon}$ thanks to Lemma~\ref{a_lem:prob_bound};
	if this occurs $k$ times,
	we have
	$\F(\Ss_\k)=\F(\Sso)\ge\F(\Sso)-\epsilon$.
	Consequently,
	\singlerun\
	returns $\Ss_\k$ that satisfies
	$\F(\Ss_\k)\ge\F(\Sso)-\epsilon$
	with a probability of
	at least $q\coloneqq\Big(\gbk\cdot
	\frac{\epsilon}{\tilde{\F}+\epsilon}\Big)^k$.
	Therefore, by setting
	$
	T
	\ge
	\Big\lceil\frac{\log\delta^{-1}}{q}\Big\rceil
	=
	\Big\lceil\big(\bgk\cdot
	\frac{\tilde{\F}+\epsilon}{\epsilon}\big)^k\log\delta^{-1}\Big\rceil	
	$,
	we can see that Algorithm~\ref{a_alg:fpt} finds a solution $\Ss$
	such that
	$\F(\Ss)\ge\F(\Sso)-\epsilon$
	with a probability of at least
	\[
	1-(1-q)^T
	\ge1-(1-q)^{-\frac{\log\delta}{q}}
	\ge 1-e^{\log \delta}=1-\delta.
	\]
	Thus, the proof is completed. 
\end{proof}

\clearpage

\section{Hardness of Improving  Approximation Ratio}\label{a_sec:hardness}
We first prove the hardness result (\Cref{thm:hard}) in \Cref{a_subsec:hardness_proof}, 
and we then discuss the hardness for some easier subclasses of WMM 
in \Cref{a_subsec:discussion}.

\subsection{Proof of~\Cref{thm:hard}}\label{a_subsec:hardness_proof}
We prove the hardness result.
As with the proof of~\citep{nemhauser1978best}, 
we design objective function $\F$ appropriately 
and show that the problem of achieving
an approximation guarantee 
that exceeds $1-e^{-\brkk}$ 
is at least as hard as 
another problem that cannot be solved in polynomial time;  
roughly speaking, 
given 
an unknown subset $\Ms$ of size $k$, 
we consider seeking $\Ss\subseteq\dset$
such that
$|\Ss\cap\Ms|\ge r+1$ and $|\Ss|\le \pk\coloneqq2k-r+1$,
where $r\le \k$ is any positive integer.

We explain how to design $\F$. 
Fix the unknown subset $\Ms\subseteq\dset$ of size $k$. 
For any $\Ss\subseteq\dset$, 
we define the function value, $\F(\Ss)$, 
so that it depends only on 
$\ns\coloneqq|\Ss|$,
$\ms\coloneqq|\Ss\cap\Ms|$,
$r$, 
and $k$. 
We denote such a function by $\Gr(\ms,\ns)$,
and
we let $\F(\Ss)\coloneqq\Gr(|\Ss\cap\Ms|,|\Ss|)=\Gr(\ms,\ns)$.
For any integers 
$m\in[0,k]$ and $n\in[0,d]$ 
such that $m \le n$, 
we define the value of 
$\Gr(m,n)$ so as to satisfy the following properties: 

\begin{description}
	\item[Property 1:]
	$\F(\cdot)=\Gr(\cdot,\cdot)$
	is monotone, 
	and its 
	SBR $\brkk$
	and 
	SPR $\prkk$ 
	satisfy 
	$\brkk=1$ and 
	$\prkk\ge\prlb=
	\frac{1}{2} - \frac{1}{2}\cdot\frac{r-1}{2k-r+1}$,
	respectively.
	\item[Property 2:]
	For any $m\in[0,r]$ and $n\in[0,\d]$,
	the value of $\Gr(m,n)$ is independent of $m$;
	i.e.,
	$\Gr(0,n)=\Gr(1,n)=\cdots=\Gr(r,n)$.
	\item[Property 3:]
	$\max_{m,n:0 \le m \le n\le k}\Gr(m,n)=\Gr(k,k)=k(k-r+1)^{k-r}$.
	\item[Property 4:]
	For any $n>\pk=2k-r+1$ and $m\in[0,k]$,
	we have
	$\Gr(m,n)=k(k-r+1)^{k-r}$.
	\item[Property 5:]
	$
	\frac{\Gr(0,k)}{\Gr(k,k)}
	=
	\frac{\Gr(1,k)}{\Gr(k,k)}
	=
	\dots
	=
	\frac{\Gr(m,k)}{\Gr(k,k)}
	=1-\left(\frac{k-r+1}{k}\right)
	\left(\frac{k-r}{k-r+1}\right)^{k-r+1}
	\eqqcolon\alpha^{r-1}_k
	$.
\end{description}
As in~\citep[Lemma 4.1]{nemhauser1978best}, 
given monotone set function $\F(\Ss)=\Gr(\ms,\ns)$ 
that satisfies Properties~2--5, 
to achieve an approximation guarantee that 
exceeds $\alpha^{r-1}_k$ 
is at least as hard as the following problem:
\begin{quote}
	For the unknown subset $\Ms\subseteq\dset$ of size $k$, 
	find $\Ss\subseteq\dset$ 
	that satisfies 
	$|\Ss\cap\Ms|\ge r+1$ and $|\Ss|\le\pk$ 
	by using the following feedback:  
	Once $\Ss$ is proposed, 
	we are informed whether or not 
	$\Ss$ satisfies 
	$|\Ss\cap\Ms|\ge r+1$ and $|\Ss|\le\pk$. 
\end{quote}
Intuitively, 
this can be proved as follows. 
From Properties~2, 3 and 5, 
if we are to achieve 
an approximation guarantee that exceeds $\alpha_k^{r-1}$, 
we need at least to find $\Ss$ such that $\ms\ge r+1$ and $\ns\le\k$ ($\le\pk$), 
while the information about $\Gr$ values is worthless as long as 
$\ms\le r$ and/or $\ns>\pk$ due to Properties~2 and~4. 
This fact connects the original maximization problem to the above problem. 
Since $\Ms$ is unknown and 
no clue can be obtained 
by examining $\Ss$ 
if it violates $|\Ss\cap\Ms|\ge r+1$ 
and/or
$|\Ss|\le\pk$, 
the above problem cannot be solved via polynomially many queries. 
More precisely, the following proposition holds 
({see, the proof of \citep[Theorem 4.2]{nemhauser1978best}}): 
\begin{aprop}\label{a_prop:nem}
	Consider the maximization problem of form 
	$\max_{\Ss:|\Ss|\le\k} \F(\Ss)$, 
	where $\F(\Ss)=\Gr(\ms,\ns)$ has monotonicity 
	and Properties 2--5.  
	For this problem, 
	to achieve an approximation guarantee that exceeds $\alpha_k^{r-1}$ requires us to evaluate $\F$
	at least $\Omegarm(\d^{r+1}/k^{2r+2})$ times. 
\end{aprop}
By using the above properties and proposition, 
we obtain the main theorem. 
\begin{thm} 
	Consider a class of problems of form
	$\max_{\Ss:|\Ss|\le\k}\F(\Ss)$ 
	that satisfies the following conditions: 
	$\F$ is monotone and 
	has SBR
	$\brkk=1$ and 
	SPR
	$\prkk\ge 1/2 - \Thetarm(1/k)
	\overset{\k\to\infty}{\longrightarrow}1/2$. 
	For this class of problems,  
	no algorithms that evaluate $\F$  
	only on polynomially many subsets can 
	achieve an approximation guarantee that exceeds 
	$1-e^{-1}=1-e^{-\brkk}$ in general. 
\end{thm}

\begin{proof}
	The proof comprises two parts: 
	(I) 
	we prove the statement by assuming that 
	there exists a function $\F(\Ss)=\Gr(\ms,\ns)$ 
	satisfying Properties 1--5, 
	and 
	(II) 
	we show how to construct such a function.  
	
	\paragraph{Proof of the Statement}
	Take $k$ to be a monotone function of $\d$ 
	that satisfies
	$\lim_{\d\to\infty}k=\infty$ 
	and 
	$k= O(\d^{\frac{1-c}{2}})$, where $c$ is any constant 
	such that $0<c<1$. 
	Thanks to Property~1 and Proposition~\ref{a_prop:nem}, 
	we have the following conditions: 
	\begin{itemize}
		\item 
		$\brkk=1$ and 
		$\prkk\ge
		\frac{1}{2} - \frac{1}{2}\cdot\frac{r-1}{2k-r+1}$. 
		\item
		To achieve an approximation guarantee 
		that is better than $\alpha_k^{r-1}$ requires 
		$\Omegarm(\d^{c(r+1)})$ times function evaluation. 
	\end{itemize}
	Since we can take 
	$r$ to be any fixed positive integer satisfying  $r\le\k=O(\d^{\frac{1-c}{2}})$, 
	we see that $\Omegarm(\d^{c(r+1)})$ is not polynomial in $\d$. 
	Furthermore, 
	we have 
	$\prkk \overset{\k\to\infty}{\longrightarrow}1/2$ 
	and 
	$\alpha_k^{r-1} \overset{\k\to\infty}{\longrightarrow} 1-e^{-1}$.  
	Hence we obtain the statement by considering $\d\to\infty$. 
	
	\paragraph{Construction of $\Gr$}
	Given any positive integer $\ell$ ($\le k$), 
	we define the following function $\Hl(m,n)$ 
	for integers $m\in[0,\ell]$ and $n\in[0,\d]$ 
	that satisfy $m\le n$: 
	\begin{align}
	\Hl(m,n)
	&\coloneqq 
	\begin{cases}
	\ell^\ell
	-
	\ell^{\ell-1}
	(\ell-m)
	\left(1- \frac{1}{\ell}\right)^{n-m}
	&\text{if $n\le\k+\ell$,}
	\\
	\ell^\ell
	&\text{otherwise.}
	\end{cases}
	\end{align}
	Note that 
	the function is non-negative and that we have 
	\begin{align}
	\Hl(0,0) = 0  
	& & 
	\text{and} 
	& & 
	\Hl(0,n) = \Hl(1,n) = \ell^{\ell}\left(1-\frac{1}{\ell} \right)^n . 
	\end{align}
	Given any integers 
	$m_1, n_1, m_2, n_2$
	such that 
	\begin{align}
	0\le m_1\le n_1, 
	& &  
	0 \le m_2 \le n_2, 
	& &  
	m_1+m_2 \le \ell, 
	& &  
	\text{and} 
	& & 
	n_1+n_2 \le \d, 
	\end{align}
	we define 
	\begin{align}
	\Hldel{m_2,n_2}{m_1,n_1}
	&\coloneqq
	\Hl(m_1+m_2, n_1+n_2) - \Hl(m_1,n_1)
	\\
	&=
	\ell^{\ell-1}
	\left( 1-\frac{1}{\ell}\right)^{n_1-m_1} 
	\left( \ell - m_1 - (\ell - m_1 - m_2) \left( 1 - \frac{1}{\ell}\right)^{n_2 - m_2} \right). 
	\end{align}
	When $(m_2,n_2) = (1,1)$ and $(0,1)$, 
	for any $m_1,n_1$ satisfying the above conditions, 
	we have
	\begin{align}
	\Hldel{1,1}{m_1,n_1}
	=
	\ell^{\ell-1} 
	\left( 1-\frac{1}{\ell}\right)^{n_1 - m_1}
	& & 
	\text{and} 
	& & 
	\Hldel{0,1}{m_1,n_1}
	=
	\ell^{\ell-1}\left(1-\frac{m_1}{\ell} \right)
	\left( 1-\frac{1}{\ell}\right)^{n_1 - m_1}, 
	\end{align} 
	respectively. 
	For later use, we prove the following lemma: 
	\begin{alem}\label{a_lem:ratio} 
		For any integers $m_1, n_1, m_2, n_2$ 
		that satisfy 
		\begin{align}\label{a_eq:mnlk}
		0\le m_1 \le n_1 \le \ell, 
		& & 
		0\le m_2 \le n_2 \le \k, 
		& & 
		m_1+m_2\le\ell, 
		& &
		\text{and}
		& &
		n_1+n_2\le\d,
		\end{align}
		we have 
		\[
		\frac
		{\Hldel{m_2,n_2}{m_1,n_1}}
		{
			m_2\times\Hldel{1,1}{m_1,n_1}
			+
			(n_2-m_2)\times\Hldel{0,1}{m_1,n_1}
		}
		\ge 
		\left(2+\frac{\k-\ell}{\ell} \right)^{-1},   
		\]
		where we regard $0/0=1$. 
	\end{alem}
	
	\begin{proof}
		We rewrite the LHS of the target inequality as follows: 
		\begin{align}
		&\frac
		{\Hldel{m_2,n_2}{m_1,n_1}}
		{
			m_2\times\Hldel{1,1}{m_1,n_1}
			+
			(n_2-m_2)\times\Hldel{0,1}{m_1,n_1}
		}
		\\
		={}&
		\frac
		{
			\ell^{\ell-1}
			\left( 1-\frac{1}{\ell}\right)^{n_1-m_1}
			\left(
			\ell-m_1
			-(\ell-m_1-m_2)
			\left( \frac{\ell-1}{\ell}\right)^{n_2-m_2}
			\right)
		}
		{
			m_2\times
			\ell^{\ell-1}
			\left( \frac{\ell-1}{\ell}\right)^{n_1-m_1}
			+
			(n_2-m_2)\times
			\ell^{\ell-1}
			\left( \frac{\ell-1}{\ell}\right)^{n_1-m_1}
			\left(1-
			\frac{m_1}{\ell}
			\right)
		}
		\\
		={}&
		\frac
		{
			\ell-m_1
			-(\ell-m_1-m_2)
			\left(1- \frac{1}{\ell}\right)^{n_2-m_2}
		}
		{
			m_2
			+
			(n_2-m_2)
			\left(1-
			\frac{m_1}{\ell}
			\right)
		}
		\\
		={}&
		\frac
		{
			1-\frac{m_1}{\ell}
			-(1-\frac{m_1}{\ell}-\frac{m_2}{\ell})
			\left(1- \frac{1}{\ell}\right)^{\ell\left( \frac{n_2}{\ell}-\frac{m_2}{\ell}\right) }
		}
		{
			\frac{m_2}{\ell}
			\frac{m_1}{\ell}
			+
			\frac{n_2}{\ell}
			\left( 1 - \frac{m_1}{\ell}\right) 
		}.
		\end{align}
		By defining 
		$x\coloneqq \frac{m_2}{\ell}$, 
		$y\coloneqq \frac{n_2}{\ell}$, 
		and 
		$z\coloneqq 1 - \frac{m_1}{\ell}$, 
		we obtain  
		\begin{align}
		\frac
		{\Hldel{m_2,n_2}{m_1,n_1}}
		{
			m_2\times\Hldel{1,1}{m_1,n_1}
			+
			(n_2-m_2)\times\Hldel{0,1}{m_1,n_1}
		}
		=
		\frac{z - (z - x)\left(1 - \frac{1}{\ell} \right)^{\ell(y-x)} }{x(1-z)+yz}, 
		\label{a_eq:xyz}
		\end{align}
		where $x,y,z$ must satisfy 
		the following inequalities from~\eqref{a_eq:mnlk}:  
		\begin{align}
		0\le z \le 1, 
		& & 
		0\le x \le y \le \frac{\k}{\ell} = 1+\frac{\k-\ell}{\ell}, 
		& & 
		x \le z,
		& & 
		\text{and} 
		& & 
		y-z \le \frac{\d}{\ell} - 1.
		\end{align}
		The RHS of~\eqref{a_eq:xyz} can be bounded 
		from below by $\left(2+\frac{\k-\ell}{\ell} \right)^{-1}$ 
		as follows: 	
		\begin{align}
		\frac{z - (z - x)\left(1 - \frac{1}{k} \right)^{k(y-x)} }
		{x(1-z)+yz}
		&\ge
		\frac{z - (z - x)e^{-(y-x)} }
		{x(1-z)+yz}
		& & \cuz{$\left(1-1/k \right)^k \le e^{-1} $}
		\\
		&\ge 
		\frac{z - (z - x)\frac{1}{1+y-x} }
		{x(1-z)+yz}
		& & \cuz{$e^{-a}\le\frac{1}{1+a}$ for $a>-1$}
		\\
		&=
		\frac{1}{1+y-x} 
		\\
		&\ge
		\left( {2+\frac{\k-\ell}{\ell}}\right)^{-1}. 
		& & \cuz{$x\ge0$ and $y\le 1+\frac{k-\ell}{\ell}$}
		\end{align}
		Thus, the lemma holds. 
	\end{proof}
	
	By using the above $\Hl(m,n)$ with $\ell=\k-r+1$, 
	we construct $\Gr(m,n)$ for any integers $m\in[0,k]$ and $n\in[0,d]$ 
	as follows: 
	\begin{align}
	\Gr(m,n)
	\coloneqq
	\begin{cases}
	n\times \Hr(0,1)
	& \text{if $0\le m \le n \le r$,}
	\\
	(r-1)\times\Hr(0,1)
	+
	\Hr(0,n-r+1)
	& \text{if $0\le m \le r$ and $r \le n \le d$,}
	\\
	(r-1)\times\Hr(0,1)
	+
	\Hr(m-r+1,n-r+1) & \text{if $r\le m \le \k$ and $r \le n \le d$.}
	\end{cases}
	\end{align}
	By using $\Gr$, 
	we define $\F(\Ss)=\Gr(\ms,\ns)$.  
	We can confirm that $\Gr$ has 
	Properties 2--5 as in the proof of \citep{nemhauser1978best}. 
	Below we show that the function has Property~1. 
	We let $\d$ satisfy $\d\ge2\k$.
	The monotonicity can be confirmed easily by examining $\Fdel{j}{\Ss}$ for 
	each case. 
	Furthermore, 
	by analogy with the proof in \citep{nemhauser1978best}, 
	we can show that $\F(\Ss)=\Gr(\ms,\ns)$ 
	is submodular over all 
	subsets of size at most $2k$: 
	I.e., 
	$\Fdel{j}{\Ss}\ge\Fdel{j}{\Ts}$ for any 
	$\Ss\subseteq\Ts$ satisfying $|\Ts|<2k$ 
	and $j\notin\Ts$. 
	This suffices to prove that $\brkk=1$ holds. 
	
	Below we prove $\prkk\ge\left( 2+\frac{r-1}{k-r+1} \right)^{-1}$. 
	Note that SPR
	can be written as
	\begin{align}
	\prkk=
	\min_{
		{\small \begin{array}{l}
			\Ls,\Ss:\Ls\cap \Ss=\emptyset, \\
			|\Ls|\le\k, |\Ss|\le \k
			\end{array}}
	}
	\frac{\Fdel{\Ss}{\Ls}}{\sum_{j\in \Ss}\Fdel{j}{\Ls}},
	\end{align}
	where we regard $0/0=1$. 
	In what follows, 
	for any disjoint $\Ls,\Ss\subseteq\dset$ of size at most $k$, 
	we consider bounding
	$\frac{\Fdel{\Ss}{\Ls}}{\sum_{j\in\Ss}\Fdel{j}{\Ls}}$ 
	from below.  
	Depending on the 
	values of 
	$\ml=|\Ls\cap\Ms|$, 
	$\ms=|\Ss\cap\Ms|$, 
	$\nl=|\Ls|$, 
	and 
	$\ns=|\Ss|$, 
	we have the following six cases.   
	We first examine each case and then show that 
	$\frac{\Fdel{\Ss}{\Ls}}{\sum_{j\in\Ss}\Fdel{j}{\Ls}}\ge\left( 2+\frac{r-1}{k-r+1} \right)^{-1}$ holds 
	for all cases. 
	
	\begin{description}
		\item[Case 1: $\ns+\nl < r$.] 
		In this case, we have 
		$
		\Fdel{\Ss}{\Ls} 
		=
		(\ns-\nl)\Hr(0,1)
		$; 
		i.e., the function is modular. 
		Therefore, we have
		$
		\frac{\Fdel{\Ss}{\Ls}}{\sum_{j\in \Ss}\Fdel{j}{\Ls}}
		=1.
		$
		\item[Case 2: $\nl< r$ and  $\ml + \ms \le r \le \ns+\nl$.] 
		In this case, we have
		\begin{align}
		\Fdel{\Ss}{\Ls}
		&=
		\Hr(0,\ns+\nl-r+1) - 
		(\nl-r+1)\Hr(0,1),
		\\
		\Fdel{j}{\Ls}
		&=
		\Hr(0,1).
		\end{align}
		Note that we have $|\Ls|\le\k$ and $|\Ss|\le\k$. 
		Due to the submodularity over all subsets of size at most $2k$, 
		the more elements $\Ls$ includes, 
		the smaller $\Fdel{\Ss}{\Ls}$ becomes, 
		which means $\Fdel{\Ss}{\Ls}$ attains its minimum 
		when $\nl=r-1$. 
		Therefore, we have 
		\begin{align}
		\frac{\Fdel{\Ss}{\Ls}}{\sum_{j\in\Ss}\Fdel{j}{\Ls}}
		\ge
		\frac{\Hr(0,\ns)}{\ns\times\Hr(0,1)}
		=
		\frac{\Hrdel{0,\ns}{0,0}}{\ns\times\Hrdel{0,1}{0,0}}.  
		\end{align}
		
		\item[Case 3: $\nl< r$ and $r \le \ms+\ml$.] 
		We have 
		\begin{align}
		\Fdel{\Ss}{\Ls} 
		&=
		\Hr(\ml+\ms-r+1,\ns+\nl-r+1) - 
		(\nl-r+1)\Hr(0,1),
		\\
		\Fdel{j}{\Ls}
		&=
		\Hr(0,1).
		\end{align}
		By analogy with the above case, 
		$\Fdel{\Ss}{\Ls}$ attains its minimum when $\nl=r-1$. 
		Therefore, 
		\begin{align}
		\frac{\Fdel{\Ss}{\Ls}}{\sum_{j\in\Ss}\Fdel{j}{\Ls}}
		&\ge 
		\frac{\Hr(\ms+\ml-r+1,\ns)}{\ns\times\Hr(0,1)}
		=
		\frac{\Hrdel{\ms+\ml-r+1,\ns}{0,0}}{\ns\times\Hrdel{0,1}{0,0}}  
		\\
		&
		=
		\frac{\Hrdel{\ms+\ml-r+1,\ns}{0,0}}
		{
			(\ms+\ml-r+1)\Hrdel{1,1}{0,0}
			+
			(\ns-\ms-\ml+r-1)\Hrdel{0,1}{0,0}
		}, 
		\end{align}
		where the last equality comes from $\Hr(0,n)=\Hr(1,n)$. 
		Note that we have 
		$\ms+\ml-r+1\le\ell=\k-r+1$ 
		since 
		$\ms+\ml=|\Ss\cap\Ms|+|\Ls\cap\Ms|\le|\Ms|=\k$, 
		where the inequality comes from the fact that $\Ls$ and $\Ss$ are disjoint. 
		Furthermore, we have 
		$\ns-\ms-\ml+r-1\ge0$ since 
		$\ns\ge\ms$ and $\ml\le\nl\le r-1$. 
		
		\item[Case 4: $\ml< r \le \nl$ and $\ms+\ml \le r$.]
		We have 
		\begin{align}
		\Fdel{\Ss}{\Ls} 
		&=
		\Hr(0,\ns+\nl-r+1) - 
		\Hr(0,\nl-r+1)
		\\
		&=
		\Hrdel{0,\ns}{0,\nl-r+1},
		\\
		\Fdel{j}{\Ls} 
		&=
		\Hrdel{0,1}{0,\nl-r+1}, 
		\end{align}
		and thus we obtain  
		\begin{align}
		\frac{\Fdel{\Ss}{\Ls}}{\sum_{j\in \Ss}\Fdel{j}{\Ls}} 
		= 
		\frac{\Hrdel{0,\ns}{0,\nl-r+1}}
		{\ns\times \Hrdel{0,1}{0,\nl-r+1}}.  
		\end{align}
		
		\item[Case 5: $\ml< r \le \nl$ and $\ms+\ml \ge r$.] 
		We have 
		\begin{align}
		\Fdel{\Ss}{\Ls} 
		&=
		\Hr(\ms+\ml-r+1,\ns+\nl-r+1) - 
		\Hr(0,\nl-r+1)
		\\
		&=
		\Hrdel{\ms+\ml-r+1,\ns}{0,\nl-r+1},
		\\
		\Fdel{j}{\Ls} 
		&=
		\Hrdel{0,1}{0,\nl-r+1}, 
		\end{align}
		and thus we obtain  
		\begin{align}
		&\frac{\Fdel{\Ss}{\Ls}}{\sum_{j\in \Ss}\Fdel{j}{\Ls}}  
		\\
		={}&
		\frac{\Hrdel{\ms+\ml-r+1,\ns}{0,\nl-r+1}}
		{\ns\times \Hrdel{0,1}{0,\nl-r+1}}. 
		\\
		={}&
		\frac{\Hrdel{\ms+\ml-r+1,\ns}{0,\nl-r+1}}
		{
			(\ms+\ml-r+1)\Hrdel{1,1}{0,\nl-r+1}
			+
			(\ns-\ms-\ml+r-1)\Hrdel{0,1}{0,\nl-r+1}
		},
		\end{align}
		where we used $\Hrdel{0,1}{0,n}=\Hr(0,n+1)-\Hr(0,n)
		=\Hr(1,n+1)-\Hr(0,n)=\Hrdel{1,1}{0,n}$. 
		Note that we can obtain 
		$\ms+\ml-r+1\le\ell=\k-r+1$ 
		and 
		$\ns-\ms-\ml+r-1\ge0$ 
		by analogy with Case~3. 
		
		\item[Case 6: $\ml\ge r$.] 
		We have
		\begin{align}
		\Fdel{\Ss}{\Ls} 
		&=
		\Hr(\ms+\ml-r+1,\ns+\nl-r+1) - 
		\Hr(\ml-r+1,\nl-r+1)
		\\
		&=
		\Hrdel{\ms,\ns}{\ml-r+1,\nl-r+1},
		\\
		\Fdel{j}{\Ls}
		&=
		\begin{cases}
		\Hrdel{0,1}{\ml-r+1,\nl-r+1} 
		& \text{if $j\notin\Ms$}, 
		\\
		\Hrdel{1,1}{\ml-r+1,\nl-r+1} 
		& \text{if $j\in\Ms$}.  
		\end{cases}
		\end{align}
		In this case, we obtain 
		\begin{align}
		&\frac{\Fdel{\Ss}{\Ls}}{\sum_{j\in \Ss}\Fdel{j}{\Ls}} 
		\\
		={}&
		\frac{\Hrdel{\ms,\ns}{\ml-r+1,\nl-r+1}}
		{\ms\times\Hrdel{1,1}{\ml-r+1,\nl-r+1}
			+
			(\ns-\ms)\times\Hrdel{0,1}{\ml-r+1,\nl-r+1} }.
		\end{align}	
	\end{description}
	In all cases, 
	the value of 
	$\frac{\Fdel{\Ss}{\Ls}}{\sum_{j\in \Ss}\Fdel{j}{\Ls}}$ 
	is lower bounded by $\left(2+\frac{r-1}{k-r+1} \right)^{-1} $ 
	thanks to Lemma~\ref{a_lem:ratio} with $\ell=k-r+1$, 
	where we let 
	$(m_1, n_1, m_2, n_2) = 
	(0, 0, 0, \ns), 
	(0, 0, \ms+\ml-r+1, \ns), 
	(0, \nl-r+1, 0, \ns), 
	(0, \nl-r+1, \ms+\ml-r+1, \ns)
	$, 
	and 
	$(\ml-r+1, \nl-r+1, \ms, \ns)$ 
	in Cases~2, 3, 4, 5, and 6, respectively.  
	Note that 
	the conditions required in Lemma~\ref{a_lem:ratio} are satisfied
	in all cases. 
\end{proof}

\subsection{Discussion on Easier Subclasses of  WMM}\label{a_subsec:discussion}
While \Cref{thm:hard} 
means that the existing ($1-e^{-\br_{\Ss,\k}}$)-approximation guarantee~\citep{das2018approximate} 
achieved by \greedy{} 
cannot be improved in polynomial time in general,  
there may exist easier subclasses of WMM 
that are not considered in \Cref{thm:hard}, 
for which better guarantees may be possible.  
One such example is 
WMM such that 
SBR and SPR defined on the whole domain
(i.e., $\br_{\d}$ and $\pr_{\d}$) are lower bounded 
by constants.   
Therefore, 
as regards algorithms whose guarantees 
are proved by 
using the (weak) submodularity on the whole domain,  
their guarantees may be improved by using 
bounded $\pr_{\d}$. 
One such algorithm is the {\it continuous greedy} algorithm~\citep{calinescu2011maximizing}, 
and so a better guarantee for WMM with 
bounded $\br_{\d}$ and $\pr_{\d}$ may be possible 
by using continuous-greedy-based methods. 
However, 
whether this approach works or not is 
non-trivial since 
we currently lack a guaranteed 
{\it rounding} scheme for WM functions. 
On the other hand, 
as in~\citep{bian2017guarantees,das2018approximate}, 
the proofs of \greedy{} 
rely only on the weak submodularity defined on 
the restricted domain (or bounded $\brkk$). 
With regard to such algorithms, 
\Cref{thm:hard} suggests that, 
even if $\br_{\d}$ and $\pr_{\d}$ are bounded, 
it is hardly possible to obtain 
approximation ratios better than $1-e^{-\brkk}$ 
only via slight modification of the existing proofs. 
To conclude, 
there are the following two possibilities. 
It is hard to 
improve $1-e^{-\brkk}$ even for easier subclasses of 
WMM (e.g., $\br_\d$ and $\pr_\d$ are bounded), 
or 
there exist new techniques 
(e.g., a continuous-greedy-based one) 
and 
it is possible to obtain approximation guarantees 
that can exceed $1-e^{-\brkk}$ 
for some easier subclasses.

\end{document}